%% file: main.tex
\pdfoutput=1
\documentclass[10pt]{llncs}
\usepackage[T1]{fontenc}
\usepackage{graphicx}
\usepackage{amssymb}
\usepackage{amsmath}
\usepackage{dcolumn}
\usepackage{bm, color}
\usepackage{graphicx}
\usepackage{diagbox} 
\usepackage{subfig}
\usepackage{caption}
\usepackage{algpseudocode}
\usepackage{algorithmicx}
\usepackage{algorithm}
\usepackage{makecell}
\usepackage{multirow}
\usepackage{hyperref}
\hypersetup{
    colorlinks=true,
    linkcolor=blue,
    filecolor=blue,      
    urlcolor=blue,
    citecolor=blue,
}

\usepackage{tikz}
\usetikzlibrary{positioning,patterns}
\usetikzlibrary{decorations.pathreplacing,decorations.pathmorphing}

\newcommand{\True}{\textbf{true}}
\newcommand{\False}{\textbf{false}}

\renewcommand{\epsilon}{\varepsilon}

\renewcommand{\dh}{\ensuremath{\mathcal{D(H)}}}
\newcommand{\doo}{\ensuremath{\mathcal{D(O)}}}

\DeclareMathOperator{\id}{id}

\renewcommand{\c}{\ensuremath{\mathcal{C}}}
\renewcommand{\u}{\ensuremath{\mathcal{U}}}
\newcommand{\e}{\ensuremath{\mathcal{E}}}

\newcommand{\h}{\ensuremath{\mathcal{H}}}

\renewcommand{\i}{\ensuremath{\mathcal{I}}}
 \renewcommand{\r}{\ensuremath{\mathcal{R}}}

\renewcommand{\o}{\ensuremath{\mathcal{O}}}
\newcommand{\m}{\ensuremath{\mathcal{M}}}
\newcommand{\mm}{\ensuremath{\mathbb{M}}}
\renewcommand{\j}{\ensuremath{\mathcal{J}}}

\renewcommand{\l}{\ensuremath{\mathcal{L}}}
\newcommand{\pr}{\ensuremath{\text{Pr}}}

 \newcommand{\tr}{{\rm tr}} 
 \renewcommand{\a}{\ensuremath{\mathcal{A}}}

\newcommand{\bra}[1]{\langle #1 |}
\newcommand{\ket}[1]{| #1 \rangle}
\newcommand{\braket}[2]{\langle #1 | #2 \rangle}
\newcommand{\ketbra}[2]{| #1 \rangle \langle #2 |}

\newcommand{\jg}[1]{\color{red} {JG: #1 :JG} \color{black} }

\newcommand{\tth}{\texttt{h}}
\newcommand{\tts}{\texttt{s}}
\newcommand{\ttm}{\texttt{m}}
\begin{document}
\title{Verifying Fairness in Quantum Machine Learning}
%
%

\author{Ji Guan\inst{1} \and
Wang Fang\inst{1,2} \and
Mingsheng Ying\inst{1,3}}
\authorrunning{J. Guan et al.}
%
\institute{State Key Laboratory of Computer Science, Institute of Software, Chinese Academy of Sciences, Beijing 100190, China \and
University of Chinese Academy of Sciences, Beijing 100049, China \and 
Department of Computer Science and Technology, Tsinghua University, Beijing 100084, China}

\maketitle              
\pagestyle{plain}

\begin{abstract}
Due to the beyond-classical capability of quantum computing, quantum machine learning is applied independently or embedded in classical models for decision making, especially in the field of finance. Fairness and other ethical issues are often one of the main concerns in decision making. In this work, we define a formal framework for the fairness verification and analysis of quantum machine learning decision models, where we adopt one of the most popular notions of fairness in the literature based on the intuition --- any two similar individuals must be treated similarly and are thus unbiased. We show that quantum noise can improve fairness and 
develop an algorithm to check whether a (noisy) quantum machine learning model is fair. In particular, this algorithm can find bias kernels of quantum data (encoding individuals) during checking. These bias kernels generate infinitely many bias pairs for investigating the unfairness of the model. Our algorithm is designed based on a highly efficient data structure --- Tensor Networks --- and implemented on Google's TensorFlow Quantum. The utility and effectiveness of our algorithm are confirmed by the experimental results, including income prediction and credit scoring on real-world data, for a class of random (noisy) quantum decision models with 27 qubits ($2^{27}$-dimensional state space) tripling ($2^{18}$ times more than) that of the state-of-the-art algorithms for verifying quantum machine learning models.
\end{abstract}
\section{Introduction}

\textbf{Quantum Machine Learning}: Google's quantum supremacy (or advantage) experiment demonstrated that a quantum computer \textit{Sycamore} with 53 noisy superconducting qubits can do a specific calculation, namely sampling, in 200 seconds that would take (arguably) 10,000 years on the largest classical computer using existing algorithms~\cite{arute2019quantum}.  
More recently, a quantum computer \textit{Jiuzhang} with 76 noisy photonic qubits was used to perform a type of Boson sampling 
in 20 seconds that would require 600 million years for a classical computer~\cite{zhong2020quantum}. These experiments mark the beginning of the Noisy Intermediate-Scale Quantum (NISQ) computing era, where quantum computers with tens-to-hundreds of qubits become a reality, but quantum noise still cannot be avoided.

Quantum machine learning is believed to be a far frontrunner in setting a path for practical beyond-classical applications of NISQ quantum devices. This stimulates the fast development of various quantum machine learning (see~\cite{biamonte2017quantum} for a review). Stepping into industries,  Google recently built up a framework \textit{TensorFlow Quantum} for the design and training of quantum machine learning within its well-known classical machine learning platform --- \textit{TensorFlow}~\cite{broughton2020tensorflow}.

Classical machine learning has led to automated decision models assuming a signiﬁcant role in making real-world decisions, especially in finance~\cite{dixon2020machine}. Such (financial) decision tasks are known to face the curse of dimensionality as there are too many features available to model customers/users. Principal component analysis (PCA) is one of the most popular methods for dimensionality reduction. It was recently shown that quantum PCA algorithm~\cite{lloyd2014quantum} can run exponentially faster on a quantum processor. At the same time, the training process of quantum machine learning could be sped up exponentially (compared with classical training) by using quantum PCA to implement iterative gradient descent methods for network training~\cite{rebentrost2019quantum}. It is worth noting that this quantum approach is generic in the sense that it can be applied to various types of neural networks, including shallow, convolutional, and recurrent networks, and thus can mitigate the high complexity issue of classical training. Because of these reasons, quantum machine learning has been introduced to be applied independently or embedded in classical decision-making models, e.g. fraud detection (in transaction monitoring)~\cite{liu2018quantum,di2021quantum}, credit assessments (risk scoring for customers)~\cite{unknown,milne2017optimal}, and recommendation systems for content dissemination~\cite{kerenidis2016quantum} (see reviews \cite{egger2020quantum,orus2019quantum} for more information). Similar to the classical counterparts, the quantum models are trained on individuals' information, e.g. saving, employment, salary (encoded as quantum data).

\textbf{Fairness in Machine Learning}: It is well-known that classical decision models are prone to discriminating against users/consumers on the basis of characteristics such as race and gender~\cite{flores2016false}, and have even led to legal mandates of ensuring \emph{fairness}. To develop fair models, various attempts have been made to precisely define and quantify fairness. They broadly fall into two categories: \textit{group} and \textit{individual} fairness. Group fairness aims to achieve through statistical parity the same outcomes across different protected groups (e.g. gender or race)~\cite{calders2009building,pedreshi2008discrimination}, whereas individual fairness advocates treating similar individuals similarly (receiving the similar outcomes)~\cite{dwork2012fairness} (see~\cite{barocas-hardt-narayanan,binns2020apparent} for various definitions of fairness and discussions about their relationship).  
The computer science community has endeavoured to check and avoid bias in classical decision models in the sense of different types of fairness~(e.g.~\cite{dwork2012fairness,barocas-hardt-narayanan,john2020verifying}). In particular, several verifiers for formal analysis and fairness verification have been designed and implemented, including FairSquare \cite{albarghouthi2017fairsquare}, VeriFair~\cite{bastani2019probabilistic} and  Justicia~\cite{ghosh2020justicia}.

Inevitably, the same issue of fairness arises in the quantum models too. Furthermore, as quantum machine learning is principled by quantum mechanics, which is usually hard to explain to the end-users, it is even more important to verify fairness when a decision is made by a quantum machine learning algorithm. However, to the best of our knowledge, the verification problem of fairness in quantum algorithms has not yet been touched. 


\textbf{Contributions of This Paper}: In this work, we define a formal framework so that the fairness of quantum machine learning decision models can be verified and analyzed in a principled way. Our \textit{design decision} is as follows: we focus on  individual fairness --- \emph{treating similar individuals similarly}~\cite{dwork2012fairness}. The trace distance --- one of the most widely used quantities in quantum information \cite[Section 9.2]{nielsen2010quantum} --- is chosen as the metric for measuring the similarity of quantum data (individuals) in defining fairness. Our main technical contributions include: 
\begin{enumerate}\item[(1)] \textit{\textbf{Problem Reduction}}: We prove that for a given (noisy) quantum decision model, checking the fairness can be reduced to a variant of distinguishing quantum measurements (states), a fundamental problem in quantum information. We resolve this specific variant problem by finding the maximum difference between the eigenvalues of the matrices generated by quantum measurements.
As a corollary, we show that quantum noise can improve fairness.
\item[(2)] \textit{\textbf{Algorithm}}: Based on (1), an algorithm is developed to exactly and efficiently check whether or not a quantum machine learning decision model is fair. A special strength of this algorithm is that it can identify \emph{bias kernels} during the checking, and these kernels generate infinitely many \emph{bias pairs}, that is, two similar quantum data that are not treated similarly. Then these bias pairs can be used to investigate the bias of the decision model.
\item[(3)] \textit{\textbf{Case Studies}}: The effectiveness of our algorithm is confirmed by experiments on quantum (noisy) decision models with 8 or 9 quantum bits (qubits) for income prediction and credit scoring on real-world data. In particular, its efficiency is shown by a class of random quantum decision models with 27 qubits, which works on a  $2^{27}$-dimensional state space.
The state-of-the-art verification algorithm~\cite{guan2021robustness} for quantum machine  learning was only able to deal with (the robustness with) 9 qubits. Our experiments can be considered a big step toward the demanded number ($\geq$ 50) of qubits in practical applications of the NISQ era. 
\end{enumerate}


\subsection{Related Works and Challenges}\label{sec:related_works}

To put our work in an appropriate context, let us further discuss some related works and the challenges we face in this paper. 

{\vskip 3pt}

\textbf{Classical versus Quantum Models:} 
In order to identify and mitigate the bias of classical machine learning decision models, 
an algorithm for maximizing utility with fairness guarantee was proposed~\cite{dwork2012fairness}. Then the strategy of searching input data with linear and integral constraints is employed in a verifier for proving individual fairness of a given decision model~\cite{john2020verifying}. The verifier is sound but not complete in general. But in the case of linear models, it is exact (both sound and complete) if the worst-case exponential time is allowed. However, although quantum decision models are always linear, the above technique cannot be directly generalized from the classical case to the quantum case. The main obstacle here is that the corresponding constraints in the quantum models are nonlinear, and thus searching the data set in a linear domain is ineffective in the quantum case. In this paper, we surmount this obstacle by reducing the quantum fairness verification problem to determining the distinguishability of a quantum measurement, which is independent of input data. Then we resolve the latter by eigenvalue analysis with polynomial time in the dimension of input quantum data. As a result, our algorithm is exact (sound and complete) and efficient.

{\vskip 3pt}

\textbf{Fairness versus Robustness:} As in the classical case, the individual fairness considered in this paper can be thought of as a kind of global robustness~\cite{john2020verifying}. This will be formally discussed in Section~\ref{Sec:definition}. In the last few years, quite a few papers have been devoted to (adversarial) robustness verification of quantum machine learning (e.g.~\cite{guan2021robustness,liu2020vulnerability,weber2020optimal}), where a verifier is given a nominal input quantum datum and it checks robustness in a  neighborhood of that particular input datum. However, the techniques developed in these works cannot be directly generalized to solve our problem of fairness verification, because we are required to check a global property. Instead, we transfer the impact of the evolution of the quantum machine learning model on input quantum data to quantum measurements.



{\vskip 3pt}

\textbf{Efficiency:} 
As the dimension of input data increases exponentially with the number of qubits, efficiency is always a key issue in the verification of quantum machine learning models. The state-of-the-art algorithms for robustness verification mentioned above can only cope with quantum machine learning models with 9 qubits\footnote{The experiments of~\cite{guan2021robustness} were performed on a personal computer and the size is at most 8 qubits. We have estimated and tested the same experiments on the server we used in this paper and only 9 qubits can be handled.}. In this paper, we boost the scale up to 27 qubits on a small server, which represents a big step toward the demand in practical applications of NISQ devices ($\geq$ 50 qubits). The speedup originates from not only the high efficiency of our algorithm but also the based data structure we adopted --- \emph{Tensor Network}~\cite{biamonte2017tensor} --- which can exploit the locality and regularity of the underlying circuits of quantum decision models and thus further optimize the algorithm. 


\section{Quantum Decision Models}\label{sec:preliminary}
For convenience of the reader, in this section, we review the setup of quantum (machine learning) decision models in their most basic form. 

{\vskip 3pt}

{\bf Classical Models:} In the classical world, a  \emph{classification decision model} is a mapping $f_c:\c\rightarrow \o$, where $\c$ is a set of data to be classified,  and $\o$ is a set of outcomes corresponding to the classes we are interested in; for example $\o=\{0,1\}$ in the simplest non-trivial (binary) case. Such a  model $f_c$ can be generalized to be a randomized mapping $f_r:\c\rightarrow \doo$, where $\doo$ denotes the set of probability distributions over $\o$.  $f_r$ is known as a \emph{regression decision model} to predict distributions and naturally describes a randomized classification procedure: to classify $x\in \c$, choose an outcome $o\in\o$ according to the distribution $f_r(x)$. For example, $o$ is chosen as  the outcome corresponding to the maximum probability of $f_r(x)$. Therefore, the basic form of a classical decision model is a randomized mapping $f=f_r$ ($f=f_c$ when $f$ is degenerated to be a deterministic mapping).

\vspace{-12pt}

\begin{figure}[ht]
    \centering
    \input{classification.tikz}
    \vspace{-0.2cm}
    \caption{Noisy Quantum (Machine Learning) Decision Model}
    \label{fig:classification}
    \vspace{-0.4cm}
\end{figure}
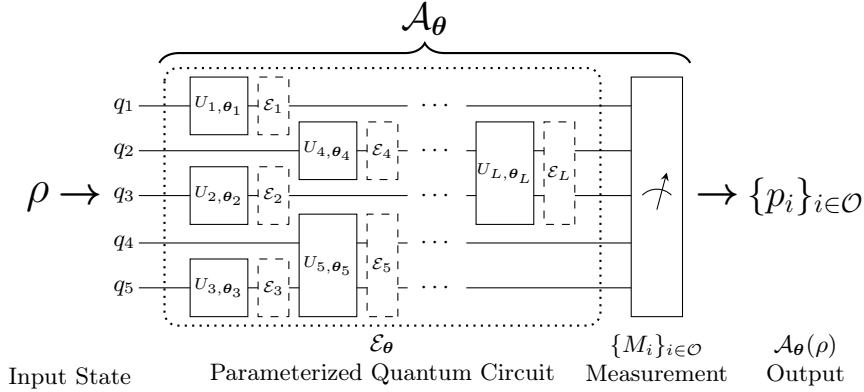

{\bf Quantum Models:}  Due to the statistical  nature of quantum mechanics, a quantum decision model is inherently a randomized mapping $\a: \dh\rightarrow \doo$. Here $\dh$ is the set of \emph{quantum states} (data) and to be specific later. Inspired by the classical models, $\a$ is not predefined but initialized as $\a_\theta$ by a parameterized quantum circuit $\e_\theta$ (see Fig.~\ref{fig:classification}) with a set of free parameters $\theta=\{\theta_{j}\}_{j=1}^{L}$. Following the training strategy of classical machine learning, $\a_\theta$ is trained on a set of input quantum states (training dataset) by tuning $\theta$ subject to some loss function $\l(\theta)$.

In the following, we explain the noisy quantum decision model from the left side to the right one of Fig.~\ref{fig:classification}. For the details of the training process, we refer to a comprehensive review paper~\cite{cerezo2021variational}.

{\vskip 3pt}

\textbf{Input State $\bm{\rho}$:} The input state of the model is a quantum \emph{mixed state} $\rho$, which is mathematically modelled by a positive semi-definite complex matrix, written as $\rho\geq 0$, with unit trace\footnote{$\rho$ has unit trace if $\tr(\rho) = 1$, where trace $\tr(\rho)$ of $\rho$ is defined as the summation of diagonal elements of $\rho$.}. $\rho$ admits a decomposition form\footnote{This kind of decomposition is generally infinitely many, and one instance is eigen-decomposition, i.e., $p_k$ and $\ket{\psi_k}$ are eigenvalues and  eigenvectors of $\rho$, respectively}:
$\rho = \sum_{k} p_{k}\psi_k$
where $\{p_{k}\}$ is a probability distribution and each $\psi_k$ is a rank-one positive semi-definite matrix, i.e., $\psi_k=\ket{\psi_k}\bra{\psi_k}$. Here, $\ket{\psi_k}$ is a unit vector and $\bra{\psi_{k}}$ is the entry-wise conjugate  transpose of $\ket{\psi_{k}}$, i.e., $\bra{\psi_{k}}=\ket{\psi_{k}}^\dagger$. Physically, $|\psi_k\rangle$ represents a \textit{pure state}, and $\rho$ represents an ensemble $\{(p_{k},\ket{\psi_{k}})\}_{k}$, often called a mixed state, meaning that $\rho$ is at $\ket{\psi_{k}}$ with probability $p_{k}$. In particular, if $\rho=\psi$ for some pure state $\ket{\psi}$, then the ensemble is deterministic; that is, it is degenerated to a singleton  $\{(1,\psi)\}$. 
In general, the statistical feature of $\rho$ may result from  quantum noise, which is unavoidable in the current NISQ era, from the surrounding environment.
\begin{example} [Qubits -- Quantum Bits] A pure state of a single qubit $q$
is described by a 
2-dimensional unit vector and in the Dirac notation it can be written as: 
\[\ket{\psi}=\left(\begin{array}{cc}a\\ b\end{array}\right)=a\ket{0}+b\ket{1}  \text{ for }\ket{0}=\left(\begin{array}{cc}1\\ 0\end{array}\right), \ket{1}=\left(\begin{array}{cc}0\\ 1\end{array}\right) \text{ and } |a|^2+|b|^2=1,\]
and ensembles $\{(\frac{1}{2},\ket{0}),(\frac{1}{2},\ket{+})\}$ and $\{(\frac{1}{6},\ket{1}),(\frac{5}{6},\ket{\phi})\}$ of $q$  are represented by the same 2-by-2 mixed state \[\rho=\frac{1}{4}\begin{pmatrix}3&1\\ 1&1\end{pmatrix}=\frac{1}{2}\ketbra{0}{0}+\frac{1}{2}\ketbra{+}{+}=\frac{1}{6}\ketbra{1}{1}+\frac{5}{6}\ketbra{\phi}{\phi},\] where $\ket{+}=\frac{1}{\sqrt{2}}(\ket{0}+\ket{1})$ and $\ket{\phi} = \frac{1}{\sqrt{10}}(3\ket{0}+\ket{1})$.


\end{example}

For a system of multiple qubits $q_1,...,q_n$, the state space is a $2^n$-dimensional Hilbert (linear) space, denoted by $\h$. As a result, pure and mixed states on $\h$ are $2^n$-dimensional unit vectors and $2^n\times 2^n$ positive semi-definite matrices with unit trace, respectively. It is worth noting that \emph{the dimension $2^n$ of the state space $\h$ of quantum states is exponentially increasing with the number $n$ of qubits}. Thus, describing a quantum system with a large number of qubits and verifying its properties on a classical computer is challenging. For our purpose of verifying fairness in quantum machine learning, we adopt a compact data structure --- \emph{Tensor Networks} --- to mitigate this issue (see this in Section~\ref{Sec:Evaluation}).



{\vskip 3pt}

{\bf Parameterized Quantum Circuit $\bm{\e_\theta}$:} Several different types of parameterized quantum circuits have been proposed; 
e.g. quantum neural networks (QNNs)~\cite{beer2020training} and quantum convolutional neural networks (QCNNs)~\cite{cong2019quantum}. 
Basically, 
$\e_{\theta}$ consists of a sequence of quantum operations: 
$\e_\theta=\e_{d,\theta_d}\circ\cdots \circ\e_{1,\theta_{1}}$. For each input quantum state $\rho$, the output of the circuit is $\e_\theta(\rho)=\e_{d,\theta_d}(\ldots \e_{2,\theta_2}(\e_{1,\theta_1}(\rho)))$. 
In the current NISQ era, 
each component $\e_{i,\theta_i}$ is:\begin{itemize}\item either a parameterized quantum gate $\u_{i,\theta_i}$ (the full boxes in Fig~\ref{fig:classification}) with $\u_{i,\theta_i}(\rho)=U_{i,\theta_i}\rho U_{i,\theta_i}^\dag$, where $U_{i,\theta_i}$ is a unitary matrix with parameters $\theta_i$, i.e., $U_{i,\theta_i}^\dagger U_{i,\theta_i}=U_{i,\theta_i} U_{i,\theta_i}^\dagger=I$ (the identity matrix), and $U_{i,\theta_i}^\dagger$ is the entry-wise conjugate  transpose of $U_{i,\theta_i}$; \item or a quantum noise $\e_i$ (the dashed boxes in Fig~\ref{fig:classification}). Mathematically, it can be described by a family of Kraus matrices $\{E_{ij}\}$~\cite{nielsen2010quantum}:
$\e_i(\rho)=\sum_j E_{ij}\rho E_{ij}^\dagger$ with $\sum_{j}E^\dagger_{ij}E_{ij}=I$. Briefly, $\e_i$ is represented as $\e_i=\{E_{ij}\}$.
\end{itemize} Note that in constructing a quantum machine learning model, only quantum gate $\u_{i,\theta_i}$ is parameterized, and noises $\e_i$ are not because they come from the outside environment. 

It should be pointed out  that, in a practical model, as shown in Fig.~\ref{fig:classification}, each quantum operation  $\e={\e_{i,\theta_{i}}}$   non-trivially applies on one or two qubits. For example, if $\e$ only works on the first qubit, then $\e=\e_1\otimes\id_2\otimes\ldots\otimes \id_n$ 
and  $\e(\rho_1\otimes \rho_2\otimes \ldots \otimes \rho_n)=\e_1(\rho_1)\otimes \rho_2\otimes \ldots \otimes \rho_n$, where $\rho_i$ is the mixed state applied on qubit $q_i$ and tensor product $\rho_1\otimes \rho_2\otimes \ldots \otimes \rho_n$ is the joint state of multiple qubits $q_1,\ldots,q_n$. This locality feature will be exploited by Tensor Networks to optimize our verification algorithm for fairness in the Evaluation Section --- Section~\ref{Sec:Evaluation}. 
\begin{example}\label{Exa:noise}
Consider the 1-qubit noise model:
 $\e_U(\rho)=(1-p)\rho+p U\rho U^\dagger$ 
 where $0\leq p\leq 1$ is a probability and  $U$ is a unitary matrix. 
 It includes the following typical noises depending on the choice of $U$: $U=X$ for bit flip, $U=Z$ for phase flip and $U=Y=\imath XZ$ for bit-phase flip~\cite[Section 8.3]{nielsen2010quantum}, where $I, X, Y, Z$ are \emph{the Pauli matrices}: 
 \[\setlength\abovedisplayskip{1.5pt}
 \setlength\belowdisplayskip{1.5pt} X=\left(\begin{matrix}0&1\\ 1&0\end{matrix}\right),\  Y=\left(\begin{matrix}0&-\imath\\ \imath&0\end{matrix}\right),\  Z=\left(\begin{matrix}1&0\\ 0&-1\end{matrix}\right),\ 
I=\left(\begin{matrix}1&0\\ 0&1\end{matrix}\right),\]
where $\imath$ denotes imaginary unit.
 The depolarizing noise combines the above three kinds of noise: 
$\e_D(\rho)=(1-p)\rho+p\frac{I}{2}=(1-\frac{3p}{4})\rho+\frac{p}{4}(X\rho X+Y\rho Y+Z\rho Z$).

\end{example}

{\vskip 3pt}

{\bf Measurement $\bm{\{M_{i}\}_{i\in\o}}$:} At the end of parameterized quantum circuit $\e_\theta$, we cannot directly read out the output $\e_\theta(\rho)$.
The only way allowed by quantum mechanics to extract classical information from $\e_\theta(\rho)$ is through a quantum measurement, which is mathematically modeled by a set $\{M_{i}\}_{i\in\o}$ of matrices 
with $\o$ being the set of possible outcomes and $\sum_{i
\in\o} M_{i}^\dagger M_{i}=I$. This observing process is probabilistic: for the measurement on  state $\e_\theta(\rho)$, an outcome $i\in\o$ is obtained with probability  
$p_{i}=\tr( M_{i}\e_\theta(\rho) M_{i}^\dagger)$\footnote{After measuring $\e_{\theta}(\rho)$ with outcome $i\in\o$, the state $\e_{\theta}(\rho)$ will be collapsed (changed) to  
$  \rho'_i={M_{i}\e_{\theta}(\rho) M_{i}^\dagger}/p_i.$
As we can see, the post-measurement  state $\rho'_i$ is dependent on the measurement outcome $i$.  
This special property is vitally different from the classical computation.}. Therefore, the output of quantum machine learning model $\a_{\theta}$ upon an input $\rho$ is a probability distribution $\a_{\theta}(\rho)=\{p_i: p_i=\tr( M_{i}\e_\theta(\rho) M_{i}^\dagger)\}$, as depicted at the rightmost  of Fig.~\ref{fig:classification}.

In this paper, we focus on the well-trained quantum machine learning models (i.e., $\theta$ has been tuned), so we ignore the $\theta$ in $\e_\theta$ and $\a_\theta$. Now, we can formally specify quantum decision model $\a$ as follows:
\begin{definition}\label{def:decision_model}
  A quantum decision model  $\a=(\e,\{M_i\}_{i\in \o})$ is a randomized mapping: 
  \[\a: \dh\rightarrow\doo\qquad \a(\rho)=\{\tr(M_i\e(\rho)M_i^\dagger)\}_{i\in\o}\quad \forall \rho\in\dh,\]
  where  $\e$ is a super-operator on Hilbert space $\h$, and $\{M_{i}\}_{i\in \o}$ is a quantum measurement on $\h$  with $\o$ being the set of measurement outcomes (classical information) we are interested in.
\end{definition}

Like their classical counterparts, quantum decision models are usually classified into two categories:   \textit{regression} and \textit{classification}  models. Regression models generally predict a value/quantity, whereas classification models predict a label/class. More specifically, 
a regression model $\a_{\r}$ uses the output of $\a$  directly as the predicted value of the regression variable $\rho\in\dh$. That is $\a_\r(\rho)=\a(\rho)$ for all $\rho\in\dh$.
In the classical world, regression models have been successfully applied to many real-world applications, such as stock market prediction and object detection.  Quantum regression models were recently used to predict molecular atomization energies~\cite{reddy2021hybrid} and the demonstration of IBM's programming platform---Qiskit~\cite[Variational Quantum Regression]{QiskitTextbook}.
On the other hand, classification model $\a_{\c}$ further uses the measurement outcome probability distribution $\a(\rho)$ to sign a class label on the input state $\rho$. The most common  way is as follows:
\[\a_{\c}:\dh\rightarrow \o \qquad \a_\c(\rho)=\arg \max_i\a(\rho)_i \quad\forall \rho\in \dh, i\in\o,\]
where $\a(\rho)_i$ denotes the $i$-th element of distribution $\a(\rho)$.
Classical classification models have broad applications in our daily life, such as face recognition and medical image classification. Quantum classification models have been used to  implement quantum phase recognition~\cite{cong2019quantum} and cluster excitation detection~\cite{broughton2020tensorflow} from real-world physical problems, and fraud detection~\cite{liu2018quantum} in finance. 

As we saw above, although classical and quantum decision models $f$ and $\a$ are both randomized mappings, the input data to them and their procedure of processing the data are fundamentally different. These differences make that the techniques for verifying classical models cannot be directly applied to quantum models and we have to develop new techniques for the latter.

\section{Defining Fairness}\label{Sec:definition}
As discussed in the Introduction, an important issue in classical machine learning is: how fair is the decision made by machines. The same issue exists for quantum machine learning. 
Intuitively, the fairness of quantum decision model  $\a$ is to treat all input states equally, i.e., there is not a pair of two closed input states that has a large difference between their corresponding outcomes.  Formally, 

\begin{definition}[Bias Pair]\label{def:bias}
  Suppose we are given a quantum decision  model  $\a=(\e,\{M_i\}_{i\in\o})$, two  distance metrics $D(\cdot,\cdot)$ and $d(\cdot,\cdot)$ on $\dh$ and $\doo$, respectively, and two small enough threshold values $1\geq \epsilon,\delta>0$. Then $(\rho,\sigma)$ is said to be an $(\epsilon,\delta)$-bias pair  if the following is true
  \begin{equation}\label{def-fair-eq} [D(\rho,\sigma)\leq \epsilon]\land [d(\a(\rho),\a(\sigma))> \delta].\end{equation}
\end{definition}

The first condition in (\ref{def-fair-eq}) indicates that the distance between input states $\rho$ and $\sigma$ is within $\epsilon$, and the second condition shows the difference between outcomes $\a(\rho)$ and $\a(\sigma)$ is beyond $\delta$. Sometimes, without any ambiguity, $(\rho,\sigma)$ is called a bias pair if $\epsilon$ and $\delta$ are preset.
\begin{definition}[Fair Model]\label{def-fair}
  Let  $\a=(\e,\{M_i\}_{i\in\o})$ be a decision model. Then  $\a$ is $(\epsilon,\delta)$-fair if there is no any $(\epsilon,\delta)$-bias pair.
\end{definition}

The intuition behind this notion of fairness is that small or non-significant perturbation of a sample $\rho$ to $\sigma$ (i.e.  $D(\rho,\sigma )\leq \epsilon$) must not be treated ``differently'' by a fair model. The choice of input distance function $D(\cdot,\cdot)$ identifies the perturbations to be considered non-significantly, while the choice of the output distance function $d(\cdot, \cdot)$ limits the changes allowed to the perturbed outputs in the model.

{\vskip 3pt}
\textbf{Fairness Implying Robustness}: As the same in the classical situation~\cite{john2020verifying}, robustness of quantum machine learning is a special case of fairness defined above. Formally, robustness is defined on a specific state $\rho$: given a quantum model $\a=(\e,\{M_i\}_{i\in\o})$, $\rho$ is $(\epsilon,\delta)$-robust if for all $\sigma\in \dh$, $D(\rho,\sigma)\leq \epsilon$ implies $d(\a(\rho),\a(\sigma))\leq \delta$. In contrast, fairness is established on all quantum states: $\a$ is $(\epsilon,\delta)$-fair if and only if $\rho$ is $(\epsilon,\delta)$-robust for all states $\rho\in\dh$. So, \emph{fairness implies robustness and can be thought of as global robustness.}

{\vskip 3pt}
\textbf{Choice of Distances}: 
The reader should have noticed that the above definition of fairness for quantum decision models is similar to that for classical decision models. But an intrinsic distinctness between them comes from the choice of distances $D(\cdot,\cdot)$  and $d(\cdot,\cdot)$. In the classical case, the distances define the similarity between individuals and their appropriate choices have been intensively discussed~\cite{dwork2012fairness}. One of the most used distances is total variation distance, measuring the closeness of individuals encoded by probability distributions. In this paper, we use it as $d(\cdot,\cdot)$ for measurement outcome distributions in Definition~\ref{def:decision_model} and choose  $D(\cdot,\cdot)$ to be the trace distance. 
Trace distance is essentially a generalization of total variation distance, and has been widely used by the quantum computation and quantum information community to define the closeness of quantum states~\cite[Section 9.2]{nielsen2010quantum}. Formally,
for two quantum states $\rho,\sigma\in \dh$,
\[D(\rho,\sigma)=\frac{1}{2}\tr(|\rho-\sigma|),\]
where $|\rho-\sigma|=\Delta_++\Delta_-$ if $\rho-\sigma=\Delta_+-\Delta_-$ with $\tr(\Delta_+\Delta_{-})=0$ and $\Delta_{\pm}$ being positive semi-definite matrix.
On the other hand, for two probability distributions $p=\{p_i\}_{i\in\o}$, $q=\{q_i\}_{i\in\o}$ over $\o$,  
$d(p,q)=\frac{1}{2}\sum_{i}|p_i-q_i|$.
In particular, for the measurement outcome distributions, we have:  
\[d(\a(\rho),\a(\sigma))=\frac{1}{2}\sum_{i}|\tr(M_i^\dagger M_i\e(\rho-\sigma) )|.\]
If $\rho$ and $\sigma$ are both diagonal matrices, i.e., $\rho=\text{diag}(p_1,\cdots,p_{|\o|})$ and $\sigma =\text{diag}(q_1,\cdots, q_{|\o|})$, then $D(\rho,\sigma)=d(p,q)$.

\section{Characterizing Fairness}

In this section, we give a characterization of fairness in terms of the Lipschitz constant and clarify its relationship with quantum noises.
\subsection{Fairness and Lipschitz Constant} The Lipschitz constant has been widely used in classical machine learning for applications ranging from robustness and fairness certification of classifiers to stability analysis of closed-loop systems with reinforcement learning controllers (e.g.~\cite{fazlyab2019efficient,szegedy2013intriguing}). In this subsection, we show that there also exists a close connection between the Lipschitz constant and fairness in the quantum setting. Let us start from an observation: 
\begin{lemma}\label{lem:Lipschitz}
      Let $\a=(\e,\{M_i\}_{i\in \o})$ be a quantum decision model. Then
      \begin{equation}
	d(\a(\rho),\a(\sigma))\leq  D(\rho,\sigma).
\end{equation}
\end{lemma}
\begin{proof}
    See Appendix~\ref{sec:proof_lip} for the proof.
\end{proof}

The above lemma indicates that quantum decision model $\a$ is automatically  $(\epsilon,\delta)$-fair whenever $\epsilon=\delta$. Furthermore, we see that $\a$ is unconditionaly \emph{Lipschitz continuous}: there exists a constant $K>0$ ($K\leq 1$ by Lemma~\ref{lem:Lipschitz}) such that for all $\rho,\sigma\in \dh$, 
\begin{equation}\label{Eq:Lipschitz}
	d(\a(\rho),\a(\sigma))\leq K D(\rho,\sigma).
\end{equation}
As usual, $K$ is called a \emph{Lipschitz constant} of $\a$. Furthermore, the smallest $K$, denoted by $K^*$, is called the (best) Lipschitz constant of $\a$. 

In the context of quantum machine learning, the following theorem shows that $K^*$ actually measures the fairness of decision model $\a$, i.e., the best (maximum) ratio of $\delta$ and $\epsilon$ in a fair model, and the states $\psi,\phi$ achieving $K^*$ can be used to find bias pairs in fairness verification.

\begin{theorem}\label{Thm:fairness_verification}\begin{enumerate}\item 
    Given a quantum decision model $\a=(\e,\{M_{i}\}_{i\in\o})$ and $1\geq\epsilon,\delta>0$, $\a$ is $(\epsilon,\delta)$-fair if and only if $\delta\geq K^*\epsilon$.
    \item If $\a$ is not $(\epsilon,\delta)$-fair, then $(\psi, \phi)$ achieving $K^\ast$ is a bias kernel; that is, for any quantum state $\sigma\in\dh$, $(\rho_\psi,\rho_\phi)$ is a bias pair where 
    \begin{equation}\label{generator}\rho_{\psi}=\epsilon\psi+(1-\epsilon)\sigma\qquad \rho_{\phi}=\epsilon\phi+(1-\epsilon)\sigma.\end{equation}\end{enumerate}
\end{theorem}
\begin{proof}[Outline]
    The \textquotedblleft if\textquotedblright\  direction of the first claim is derived by the  definitions of $(\epsilon,\delta)$-fairness and $K^*$ together with  (\ref{Eq:Lipschitz}). The \textquotedblleft only if\textquotedblright\  direction of the first claim and the second claim are both based on the existence of pure states $\ket{\psi}$ and $\ket{\phi}$ achieving $K^*$: $
      d(\a(\psi),\a(\phi))=K^* D(\psi,\phi)$. 
    The detailed proof is presented in  Appendix~\ref{sec:proof_verification}.
\end{proof}

\subsection{Fairness and Noises}\label{sec:noise_fairness}
 In this subsection, we turn to consider the relation between fairness and noise. Let us first examine a simple example.
 Assume a
 noiseless quantum decision model $\a=(\u,\{\m_{i}\}_{i\in\o})$ where $\u$ is a unitary operator, i.e., $\u=\{U\}$ for some unitary matrix $U$. 
The 1-qubit depolarizing noise in Example~\ref{Exa:noise} can be  generalized to a large-size system  with the following form: \[\e(\rho)=(1-p)\rho+p\frac{I}{N} \quad\forall \rho\in\dh,\]
 where  $0\leq p\leq 1$ and $N$ is the dimension of the state space $\h$ of the system.
 By introducing it into $\a$,  
 we obtain a noisy model  $\a_\e=(\e\circ\u,\{\m_{i}\}_{i\in\o})$. Let $K^*$ and $K^*_\e$ be the Lipschitz constants of $\a$ and $\a_\e$, respectively. A calculation (with the help of  Theorem~\ref{Thm:main} below) yields:  \begin{equation}\label{improve-K}K^*_{\e}=(1-p)K^*.\end{equation}
 Theorem \ref{Thm:fairness_verification} indicates that the less the Lipschitz constant is, the fairer the quantum machine learning model will be. So, depolarizing noise improves fairness by the order of $(1-p)$. By the way, it was shown   in~\cite{du2020quantum} that depolarizing noise can improve the robustness of quantum machine learning. This result can be strengthened by using (\ref{improve-K}) to quantitatively characterize the robustness improvement.

 
 The observation in the above example can actually be generalized to the following: 

 \begin{theorem}\label{lem:noise_increase_fairness}
      Let $\a=(\u,\{M_{i}\}_{i\in\o})$ be a quantum decision model. Then for any quantum noise represented by a super-operator $\e$, we have $K^*_{\e}\leq K^*$,
      where $K^*$ and $K^*_\e$ are  the Lipschitz constants of $\a$ and $\a_\e=(\e\circ\u,\{\m_{i}\}_{i\in\o})$.
 \end{theorem}
 \begin{proof}[Outline]
    The proof of this theorem mainly depends on  the observation that the range of $\a_\e$ is a subset of the range of $\a$, i.e.   $\{\e\circ\u(\rho):\rho\in\dh\}\subseteq \{\u(\rho):\rho\in\dh\}=\dh$. Subsequently, by Definition~\ref{def:bias} of fairness, the output distributions of $\a_\e$ are contained in that of $\a$. A restatement of this theorem in terms of quantum states (measurements) distinguishability and its full proof is presented 
    in Appendix~\ref{Measurement}.
 \end{proof}
 
 \begin{remark} The above theorem indicates that adding noises at the end of noiseless computation can always improve fairness. Indeed, this is also true when the noises appear in the middle (after any gate in the circuit). 
 \end{remark}

\section{Fairness Verification}
In this section, we develop an algorithm for the fairness verification of quantum decision models based on the theoretical results obtained in the last section. Formally, the major problem concerned in this paper is the following:
\begin{problem}[Fairness Verification Problem]\label{problem:fairness}
Given a quantum decision model $\a$ and $1\geq \epsilon,\delta>0$, check whether or not $\a$ is $(\epsilon,\delta)$-fair. If not then (at least) one bias pair $(\rho,\sigma)$ is provided.
\end{problem}

\subsection{Computing the Lipschitz Constant} 

First of all, we note that essentially, Theorem~\ref{Thm:fairness_verification} gives a verification condition for fairness in terms of the Lipschitz constant $K^*$. 
Therefore, computing $K^*$ is crucial for  fairness verification. 
However, this problem is much more difficult than that in the classical counterpart as discussed in Subsection~\ref{sec:related_works}. The following theorem provides a method to compute the Lipschitz constant $K^*$ by evaluating the eigenvalues of certain matrices.  
\begin{theorem}\label{Thm:main}\begin{enumerate}
\item	Given a quantum decision model $\a=(\e,\{M_{i}\}_{i\in\o})$. The Lipschitz constant $K^*$ is: 
	\[K^*=\max_{A\subseteq \o}[\lambda_{\max}(M_A)-\lambda_{\min}(M_A)]\text{ with }M_A=\sum_{i\in A}\e^\dagger(M_{i}^\dagger M_i),\]
	  where $\e^\dagger$ is the conjugate map\footnote{$    \e^\dagger(\rho)=\sum_{j\in\j}E_j^\dagger\rho E_j$ if $\e$ admits Kraus matrix form  $\e(\rho)=\sum_{j\in\j}E_j\rho E_j^\dagger$.} of $\e$,  and $\lambda_{\max}(M_A)$ and $\lambda_{\min}(M_A)$ are the maximum and minimum eigenvalues of positive semi-definite  matrix $M_A$, respectively. 
\item	  Furthermore, let $A^*\subseteq \o$ be an optimal solution of reaching the Lipschitz constant, i.e.,
	  \[A^*=\arg\max_{A\subseteq \o}[\lambda_{\max}(M_A)-\lambda_{\min}(M_A)]\]
	  and $\ket{\psi}$ and $\ket{\phi}$  be two normalized eigenvectors corresponding to the maximum and minimum eigenvalues of $M_{A^*}$, respectively. Then we have 
      \[
      d(\a(\psi),\a(\phi))=K^* D(\psi,\phi)=K^*,
      \]
      where $\psi=\ketbra{\psi}{\psi}$ and $\phi=\ketbra{\phi}{\phi}$.\end{enumerate}
\end{theorem}
\begin{proof}[Outline]
This theorem can be proved by reducing the problem of calculating the Lipschitz constant to determining the distinguishability of a quantum measurement. Then we claim that the distinguishability is the maximum difference between the eigenvalues of the matrices generated by the measurement. The  details are quite involved, and we postpose them into Appendix~\ref{Measurement}.\end{proof}

Based on the above theorem, we are able to develop Algorithm \ref{Algorithm:K} for computing the Lipschitz constant $K^\ast$. The correctness and complexity are provided in the next subsection.
\begin{algorithm}[ht]
\caption{Lipschitz($\a$)}
\label{Algorithm:K}
    \begin{algorithmic}[1]
    \Require A quantum decision model $\a=(\e=\{E_j\}_{j\in\j},\{M_{i}\}_{i\in\o})$  on a Hilbert space $\h$ with dimension $N$.
    \Ensure The Lipschitz constant $K^*$ and $(\psi,\phi)$ as in Theorem~\ref{Thm:main}. 
    \ForAll{$i\in \o$}
    \State $W_i=\e^\dagger(M_{i}^\dagger M_i)=\sum_{j\in\j}E_j^\dagger M_{i}^\dagger M_i E_j $
    \EndFor
    \State $K^*=0$, $A^*=\emptyset$ be an empty set and $M_{A^*}={\bf 0}$, zero matrix.
    \ForAll{$A\subseteq\o$}
    \State $M_A=\sum_{i\in A}W_i$ and $K_A=\lambda_{\max}(M_A)-\lambda_{\min}(M_A)$
    \If{$K_A>K^*$}
    \State $K^*=K_A$, $A^*=A$ and $M_{A^*}=M_A$
    \EndIf
    \EndFor
    \State  $\ket{\psi}$ and $\ket{\phi}$ are obtained two normalized eigenvectors corresponding to the maximum and minimum eigenvalues of $M_{A^*}$, respectively.
    \State\Return $K^*$ and $(\psi,\phi)$
    \end{algorithmic}
\end{algorithm}

\subsection{Fairness Verification Algorithm}

Now we are ready to present our main algorithm --- Algorithm~\ref{Algorithm} --- for  verifying fairness of quantum decision models. 
\begin{algorithm}[ht]
\small
\caption{FairVeriQ($\a,\epsilon,\delta$)}
\label{Algorithm}
    \begin{algorithmic}[1]
    \Require A quantum decision model $\a=(\e=\{E_j\}_{j\in\j},\{M_{i}\}_{i\in\o})$  on a Hilbert space $\h$ with dimension $N$, and real numbers $1\geq\epsilon,\delta>0$.
    \Ensure \True{} indicates $\a$ is $(\epsilon,\delta)$-fair or \False{} with a bias kernel pair $(\psi,\phi)$ indicates  $\a$ is not $(\epsilon,\delta)$-fair.
    \State\label{line:call} $(K^*,(\psi,\phi))$=Lipschitz($\a$)\Comment{Call Algorithm~\ref{Algorithm:K}}
    \If{$\delta\geq K^*\epsilon$}\label{line:compare}
    \State \Return \True{}
    \Else
    \State \Return \False{} and $(\psi,\phi)$
    \EndIf\label{line:the_end}
    \end{algorithmic}  
\end{algorithm}

To see the correctness of Algorithm~\ref{Algorithm}, let us first  note that the second part of Theorem \ref{Thm:main} shows that $K^*$ can be achieved by $d(\a(\psi),\a(\phi))$ for two mutually orthogonal quantum (pure) states $\psi$ and $\phi$. On the other hand, the second part of Theorem \ref{Thm:fairness_verification} asserts that such states $\psi$ and $\phi$ form a bias kernel.      
Moreover, since state $\sigma\in\dh$ in (\ref{generator}) is arbitrary and $\dh$ is an infinite set, infinitely many bias pairs can be generated from this kernel.

To analyze the complexities of Algorithm~\ref{Algorithm} and its subroutine --- Algorithm~\ref{Algorithm:K},
we first see by Theorem~\ref{Thm:fairness_verification} that for evaluating the $(\epsilon,\delta)$-fairness of quantum decision model $\a$,  the Lipschitz constant $K^*$ is sufficient and necessary. Thus the first step (Line~\ref{line:call}) of Algorithm~\ref{Algorithm} is to call Algorithm~\ref{Algorithm:K} to compute $K^*$ by the  mean of Theorem~\ref{Thm:main}. The complexity of Algorithm~\ref{Algorithm:K} mainly attributes to computing $W_i=\sum_{j\in\j}E_j^\dagger M_{i}^\dagger M_i E_j $ for each $i\in\o$, and for each $A \subseteq \o$,  $\sum_{i\in A}W_i$  and its maximum and minimum eigenvalues  (and the corresponding eigenvectors for $A=A^*$ at the end). The former calculation needs $O(N^5)$ as the multiplication of $N\times N$ matrices needs  $O(N^3)$ operations, and the number $|\j|$ of the Kraus operators $\{E_j\}_{j\in\j}$ of $\e$ can be at most $N^2$~\cite[Chapter 2.2]{wolf2012quantum}; the complexity of the latter one is $O(2^{|\o|}|\o|N^2)$ since the number of subsets of $\o$ is $2^{|\o|}$, $|A|\leq |\o|$ for any $A\subseteq \o$ and computing maximum and minimum eigenvalues with corresponding eigenvectors of $N\times N$ matrix costs $O(N^2)$. Therefore, the total complexity of Algorithm~\ref{Algorithm:K} is $O(N^5+2^{|\o|}|\o|N^2)$. After that, in Lines~\ref{line:compare}-\ref{line:the_end}, we simply compare $\delta$ and $K^*\epsilon$ to answer the fairness verification problem. So, Algorithm~\ref{Algorithm} shares the same complexity with  Algorithm~\ref{Algorithm:K}.
\begin{theorem}
The worst case complexities of Algorithms~\ref{Algorithm:K} and \ref{Algorithm} are both $O(N^5+2^{|\o|}|\o|N^2)$, where $N$ is the dimension of input Hilbert state space $\h$ and $|\o|$ is the number of the measurement outcome set $\o$.
\end{theorem}

Like their classical counterparts, quantum machine learning  models usually downscale large-dimension input data to small-size outputs. This means that the number $|\o|$ of the measurement outcome set $\o$ is far smaller than the dimension $N$ of input Hilbert state space $\h$. It is even a constant 2 in most real-world tasks for binary decisions/classifications, such as income prediction and credit scoring (see the examples in  Section~\ref{Sec:Evaluation}), and in this case, the complexities of Algorithms~\ref{Algorithm:K} and \ref{Algorithm} are both $O(N^5)$. However, the dimension $N$ is exponential in the number $n$ of the input qubits, i.e., $N=2^n$. Thus the complexity turns out to be  $O(2^{5n})$. In verification of classical models, this \emph{state-space explosion problem}~\cite{baier2008principles} can be mitigated by using some custom-made data structures to capture the features of the underlying data, e.g. Binary Decision Diagrams (BDDs)~\cite{akers1978binary}. In the quantum case, we cross this difficulty by employing a quantum data structure --- \emph{Tensor Networks (TNs)}, originating from quantum many-body physics --- to exploit the locality and regularity of the circuits representing quantum machine learning models. As a result, quantum models with up to $n=27$ qubits can be handled by our verification algorithm. 

\section{Evaluation}\label{Sec:Evaluation}

In this section, we evaluate the efficiency of our verification algorithm~(Algorithm~\ref{Algorithm:K}) on noisy quantum decision models. The algorithm is implemented on \emph{TensorFlow Quantum}~\cite{broughton2020tensorflow} --- a platform of Google for designing and training quantum machine learning algorithms. Then we test it by verifying the fairness of two groups of examples:\begin{itemize}\item \textit{Small-scale models trained from real-world data}  (Subsection~\ref{subsec:financial}): There is still no public benchmarks for quantum decision models.  We choose two publicly available financial datasets, \emph{German Credit Data}~\cite{german_credit} and \emph{Adult Income Dataset} from \emph{Diverse Counterfactual Explanations Dataset}~\cite{dice2020} and train small-scale quantum models from them on TensorFlow Quantum.
Then we evaluate the Lipschitz constant $K^*$ of the trained models by Algorithm~\ref{Algorithm:K}.

\item \textit{Medium-scale models} (Subsection~\ref{subsec:nisq}):  
Medium-scale models (10-30 qubits) are difficult to be trained on TensorFlow Quantum with a personal computer or a small server since the simulated quantum noises lead to large-size (up to $2^{30}\times 2^{30}$) matrix manipulations. Thus we turn to using a model from the tutorial of TensorFlow Quantum as a seed to generate a group of medium-scale models.
The efficiency of our algorithm is then demonstrated on these models with randomly sampled parameters.
\end{itemize}
All source codes can be found at: \url{https://github.com/Veri-Q/Fairness}.
All our experiments are carried out on a server with
Intel Xeon Platinum 8153 @ 2.00GHz $\times$ 256 Processors, 2048 GB Memory and no dedicated GPU. The machine runs Centos 7.7.1908 and each experiment is run with at most 80 processors. We use the \texttt{NumPy} and Google  \texttt{TensorNetwork}~\cite{roberts2019tensornetwork} Python packages to compute Lipschitz constants and bias kernels for small-scale models and medium-scale models, respectively. These two packages have their own advantages in different sizes. 

\subsection{A Practical  Application in Finance}\label{subsec:financial}


\begin{description}
    \item[Adult Income Dataset.] The original version of this dataset is extracted from the 1994 Census database by Barry Becker~\cite{adult_dataset}.
    We use the modified version of the adult income dataset by DiCE~\cite{dice2020}.
    Each individual in this modified dataset has $8$ features and the classification whether the income exceeds $\$50,000/\text{year}$ or not.
    We randomly select $1,000$ and $400$ data from the training dataset and test dataset contained in this modified dataset, respectively.
    The task of the quantum decision model task is to predict whether an individual’s income exceeds $\$50,000/\text{year}$ or not.
    \item[German Credit Dataset.] This dataset contains $1,000$ loan applicants with $20$ features and the classification whether they are considered as having good credit risk or not (Creditability).
    It provides $500$ applicants for the training and $500$ applicants for the test.
    By using the $p$-value with creditability for each variable~\cite{AnalysisGerman}, we have $9$ features (e.g., Account Balance, Payment Status) left as significant predictors.
    The task of the quantum model to be trained is to classify whether the person has good credit risk or not.
\end{description}

These datasets contain some categorical features, which are transformed into different integer numbers for further operations. Then we have $n\in\{8,9\}$ numbering features in total and use the following data-encoding feature map:
\[\vec{x} = (x_1,x_2,\ldots,x_n) \mapsto \ket{\psi(\vec{x})} =\bigotimes_{j=1}^nX^{x_j}\ket{0} \]
\text{for Pauli matrix $X$ defined in Example~\ref{Exa:noise}} to encode an $n$-dimensional feature vector $\vec{x}$ (each dimension is normalized by its maximum value) to an $n$-qubits quantum state $\psi(\vec{x}) = \ketbra{\psi(\vec{x})}{\psi(\vec{x})}$. 
\vspace{-0.4cm}
\subsubsection{Models:}
For the quantum decision model, we choose the basic rotation and entangling building blocks~\cite{doi:10.1126/sciadv.aaw9918} to construct  parameterized quantum circuits (see Fig.~\ref{fig:rotate_model}).
In the rotation block, without any ambiguity, we directly use $X$ and $Z$ to represent parameterized $X$-rotation $e^{-\imath\frac{\theta_1}{2}X}$ and parameterized $Z$-rotation $e^{-\imath\frac{\theta_2}{2}Z}$ on one qubit, respectively. It is worth noting that
the parameterized ($Z$-$X$-$Z$)-rotation induces universal gates on each qubit~\cite[Theorem 4.1]{nielsen2010quantum}, and thus the expressiveness of the models on one qubit is ensured.
In the entangling block, $XX$ stands for the parameterized ($X\otimes X$)-rotation $e^{-\imath\frac{\theta_3}{2}X\otimes X}$ on two qubits.
The entangling block can create entanglement between each qubit. Here entanglement is a unique feature of quantum models to express the interactions of qubits. The model is constructed by alternately using these two blocks with a quantum measurement $M$ at the end of the model. 
\begin{figure}[ht]
    \centering
    \vspace{-0.2cm}
    \input{rotation.tikz}
    \vspace{-0.3cm}
    \caption{Parameterized Quantum Circuits for Quantum Finance Decision Models.}
    \label{fig:rotate_model}
    \vspace{-0.4cm}
\end{figure}
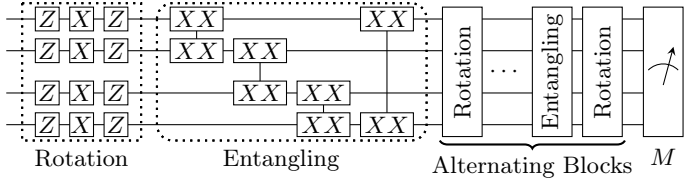

Since TensorFlow Quantum is inefficient in training noisy models, we only use $3$ rotation blocks and $2$ entangling blocks in the training models.
In addition, to simulate noisy models, we put different quantum noises introduced in Example~\ref{Exa:noise} on each qubit, including bit flip, phase flip, depolarizing, and the mixtures of them, behind the first rotation block.
Note that the number of qubits for the models is the same as the number of features of datasets due to the above choice of the data-encoding feature map.
The final measurement $M = \{M_{0} = I\otimes\ketbra{0}{0}, M_{1}=I\otimes\ketbra{1}{1}\}$ is a local measurement performed on the last qubit.
With the binary classification task, the loss function we choose is binary cross-entropy:
$-\frac{1}{N}\sum_{j=1}^N c_j \cdot \log \bar{c}_j +(1-c_j)\log(1-\bar{c}_j)),$
where $N$ is the size of the batch fixed in the training process, $c_j$ is the true label and $\bar{c}_j$ is the outcome of the measurement. All models are well trained and achieve around 70\% train and test accuracy (see Column ``Accuracy'' in Table~\ref{tab:financial}), matching that of the previously used classical and quantum finance decision models (e.g.~\cite{unknown,john2020verifying}). 

    \begin{table}[t]
        \centering
        \caption{Experimental results of Lipschitz constant $K^*$ of the  trained models.}\label{tab:financial}
        \scalebox{0.8}{\begin{tabular}{cccp{30pt}<{\centering}p{30pt}<{\centering}p{75pt}<{\centering}p{30pt}<{\centering}}
            \Xhline{1pt}
            \multirow{2}{*}{Dataset} & \multicolumn{2}{c}{Noise} & \multicolumn{2}{c}{Accuracy} & \multirow{2}{*}{$K^*$} & \multirow{2}{*}{Time} \\
            \cline{4-5} \cline{2-3}
            & type & probability & train & test \\
            \Xhline{0.6pt}
            \multirow{13}{*}{German Credit} & \multicolumn{2}{c}{None} & $0.732$ & $0.686$ & $1.0000 \times 10^{0\phantom{-} }$ & $\backslash$ \\
            \cline{2-7}
            & \multirow{3}{*}{Phase flip} & $10^{-4}$ & $0.726$ & $0.692$ & $9.9997 \times 10^{-1}$ & 2.36\tts \\
            & & $10^{-3}$  & $0.724$ & $0.714$ & $9.9800 \times 10^{-1}$ & 2.02\tts \\
            & & $10^{-2}$  & $0.704$ & $0.708$ & $9.6918 \times 10^{-1}$ & 1.94\tts\\
            \cline{2-7}
            & \multirow{3}{*}{Depolarizing} & $10^{-4}$ & $0.709$ & $0.686$ & $9.9977\times 10^{-1}$ &  2.77\tts\\
            & & $10^{-3}$  & $0.701$ & $0.712$ & $9.9789\times 10^{-1}$ & 2.93\tts \\
            & & $10^{-2}$  & $0.709$ & $0.682$ & $9.7916\times 10^{-1}$ & 3.44\tts\\
            \cline{2-7}
            & \multirow{3}{*}{Bit flip} & $10^{-4}$ & $0.712$ & $0.728$ & $9.9975\times 10^{-1}$ & 2.27\tts\\
            & & $10^{-3}$ & $0.710$ & $0.690$ & $9.9743\times 10^{-1}$ & 2.47\tts \\
            & & $10^{-2}$ & $0.724$ & $0.678$ & $9.7981\times 10^{-1}$ & 2.05\tts \\
            \cline{2-7}
            & \multirow{3}{*}{Mixed noise} & $10^{-4}$& $0.710$ & $0.704$ & $9.9980\times 10^{-1}$ & 2.15\tts\\
            & & $10^{-3}$ & $0.731$ & $0.682$ & $9.9834\times 10^{-1}$ & 2.08\tts \\
            & & $10^{-2}$ & $0.731$ & $0.692$ & $9.7021\times 10^{-1}$ & 1.95\tts\\
            \Xhline{0.6pt}
            \multirow{13}{*}{\makecell[c]{Adult Income \\ (DiCE)}} & \multicolumn{2}{c}{None} & $0.777$ & $0.770$ & $1.0000\times 10^{0\phantom{-}}$ & $\backslash$\\
            \cline{2-7}
            & \multirow{3}{*}{Phase flip} & $10^{-4}$ & $0.784$ & $0.767$ & $9.9992\times 10^{-1}$ & 0.44\tts\\
            & & $10^{-3}$ & $0.771$ & $0.770$ & $9.9805\times 10^{-1}$ & 0.51\tts\\
            & & $10^{-2}$ & $0.773$ & $0.767$ & $9.8057\times 10^{-1}$ & 0.48\tts\\
            \cline{2-7}
            & \multirow{3}{*}{Depolarizing} & $10^{-4}$ & $0.774$ & $0.767$ & $9.9987\times 10^{-1}$ & 0.57\tts\\
            & & $10^{-3}$ & $0.781$ & $0.767$ & $9.9867\times 10^{-1}$ & 0.58\tts \\
            & & $10^{-2}$ & $0.779$ & $0.767$ & $9.8667 \times 10^{-1}$ & 0.69\tts \\
            \cline{2-7}
            & \multirow{3}{*}{Bit flip} & $10^{-4}$ & $0.780$ & $0.767$ & $9.9980\times 10^{-1}$ & 0.57\tts\\
            & & $10^{-3}$ & $0.777$ & $0.767$ & $9.9800 \times 10^{-1}$ & 0.49\tts \\
            & & $10^{-2}$ & $0.778$ & $0.770$ & $9.8117 \times 10^{-1}$ & 0.54\tts \\
            \cline{2-7}
            & \multirow{3}{*}{Mixed noise} & $10^{-4}$ & $0.762$ & $0.720$ & $9.9987\times 10^{-1}$ & 0.68\tts \\
            & & $10^{-3}$ & $0.752$ & $0.720$ & $9.9812\times 10^{-1}$ & 0.67\tts \\
            & & $10^{-2}$ & $0.759$ & $0.720$ & $9.7647\times 10^{-1}$ & 0.67\tts \\
            \Xhline{0.6pt}
        \end{tabular}}
        \vspace{-0.6cm}
    \end{table}
    
\vspace{-0.4cm}
\subsubsection{Evaluation Details and Results:}
The results of evaluating Algorithm~\ref{Algorithm:K} on the models trained from different datasets and different quantum noises are presented in Table~\ref{tab:financial}.
For different datasets, we train noise-free models to serve as the baseline for training and test accuracy (see Row ``None'').
Furthermore, different types of noise are added with different levels of probabilities.
We list the Lipschitz constant $K^*$ and the running time of Algorithm~\ref{Algorithm:K} aided by \texttt{NumPy} for each column.
It can be seen that the higher level of noise's probability, the smaller value of constant $K^*$. Therefore,  the claim of quantum noise improving fairness  in~Section~\ref{sec:noise_fairness} is confirmed by the numerical results. This is also observed in Table~\ref{tab:qcnn} later.

\subsection{Scalability in the NISQ era}
\label{subsec:nisq}

\subsubsection{Models:}  
To reflect an actual application in the NISQ era, we choose not to randomly generate a parameterized quantum circuit model. Instead, we expanded the existing example of Quantum Convolutional Neural Network (QCNN)~\cite{cong2019quantum} in the QCNN tutorial\footnote{\url{https://tensorflow.google.cn/quantum/tutorials/qcnn}} of TensorFlow Quantum from $8$ qubits (see Fig.~\ref{fig:qcnn}) to $27$ qubits. In the experiment, we use the QCNN model with one convolution layer and one pooling layer.
The noise is applied between convolution and pooling layers on each qubit. The final measurement is $M = \{M_{0} = I\otimes\ketbra{0}{0}, M_{1}=I\otimes\ketbra{1}{1}\}$ performed on the last qubit with a gate $U$ appended before.
Since training a noisy model of this size is currently intractable on TensorFlow Quantum, the parameters in the model are all randomly sampled.

\begin{figure}[ht]
    \centering
    \vspace{-0.5cm}
    \input{qcnn.tikz}
    \vspace{-0.2cm}
    \caption{The QCNN model in the tutorial of TensorFlow Quantum. Each $C_i$ in the convolution layer is a parameterized 2-qubit gate to find a new state between adjacent qubits. Each $P_i$ in the pooling layer is also a parameterized 2-qubit gate with another form that attempts to extract the information of two qubits into a single qubit.}
    \label{fig:qcnn}
    \vspace{-0.4cm}
\end{figure}
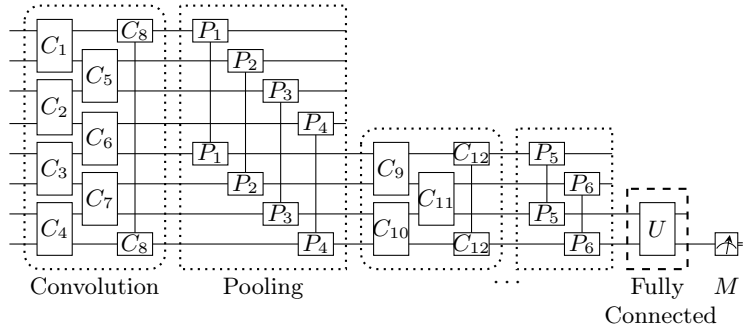

\begin{table}[ht]
    \centering
    \caption{Experimental results of Lipschitz constant $K^*$ of QCNN models.}
    \label{tab:qcnn}
    \scalebox{0.8}{\begin{tabular}{p{36pt}<{\centering}ccp{36pt}<{\centering}p{28pt}<{\raggedleft}|p{36pt}<{\centering}p{28pt}<{\raggedleft}|p{36pt}<{\centering}p{28pt}<{\raggedleft}}
            \Xhline{1pt}
            \multirow{2}{*}{\#Qubits} & \multicolumn{2}{c}{Noise} & \multicolumn{2}{c|}{Evaluation I} & \multicolumn{2}{c|}{Evaluation II} & \multicolumn{2}{c}{Evaluation III}\\
            \cline{2-9}
            & type & probability & {$K^*$} & \multicolumn{1}{c|}{Time} & {$K^*$} & \multicolumn{1}{c|}{Time} & {$K^*$} & \multicolumn{1}{c}{Time}\\
            \Xhline{0.6pt}
            \multirow{13}{*}{$25$} & \multicolumn{2}{c}{None}  & $1.0000$ & \multicolumn{1}{c}{$\backslash$} & $1.0000$ & \multicolumn{1}{c}{$\backslash$} & $1.0000$ & \multicolumn{1}{c}{$\backslash$}\\
            \cline{2-9}
            & \multirow{3}{*}{Phase flip} & $10^{-4}$  & $0.9998$ & 2.15\ttm & $0.9997$ & 1.92\ttm & $0.9999$ & 2.12\ttm \\
            & & $10^{-3}$  & $0.9983$ & 1.71\ttm & $0.9982$ & 1.35\ttm & $0.9987$ & 1.10\ttm \\
            & & $10^{-2}$  & $0.9865$ & 1.75\tth & $0.9870$ & 54.49\ttm & $0.9831$ & 39.07\ttm \\
            \cline{2-9}
            & \multirow{3}{*}{Depolarizing} & $10^{-4}$ & $0.9998$ & 2.22\ttm & $0.9998$ & 1.59\ttm & $0.9998$ & 2.38\ttm\\
            & & $10^{-3}$ & $0.9985$ & 2.46\ttm & $0.9980$ & 1.62\ttm & $0.9982$ & 2.04\ttm\\
            & & $10^{-2}$ & $0.9824$ & 2.33\ttm & $0.9802$ & 2.53\ttm & $0.9809$ & 1.77\ttm \\
            \cline{2-9}
            & \multirow{3}{*}{Bit flip} & $10^{-4}$ & $0.9997$ & 1.74\ttm & $0.9998$ & 1.60\ttm & $0.9999$ & 2.15\ttm\\
            & & $10^{-3}$ & $0.9986$ & 2.44\ttm & $0.9980$ & 1.80\ttm & $0.9991$ & 2.37\ttm \\
            & & $10^{-2}$ & $0.9943$ & 1.78\tth & $0.9854$ & 20.78\ttm & $0.9919$ & 49.36\ttm\\
            \cline{2-9}
            & \multirow{3}{*}{Mixed noise} & $10^{-4}$ & $0.9998$ & 3.68\ttm & $0.9998$ & 1.34\ttm & $0.9998$ & 1.94\ttm \\
            & & $10^{-3}$ & $0.9980$ & 1.66\ttm & $0.9966$ & 2.06\ttm & $0.9983$ & 0.96\ttm\\
            & & $10^{-2}$ & $0.9901$ & 37.24\ttm & $0.9861$ & 1.95\tth & $0.9759$ & 6.03\ttm\\
            \Xhline{0.6pt}
            \multirow{13}{*}{$27$} & \multicolumn{2}{c}{None} & $1.0000$ & \multicolumn{1}{c}{$\backslash$} & $1.0000$ & \multicolumn{1}{c}{$\backslash$} & $1.0000$ & \multicolumn{1}{c}{$\backslash$} \\
            \cline{2-9}
            & \multirow{3}{*}{Phase flip} & $10^{-4}$ & $0.9999$ & 6.75\ttm & $0.9998$ & 7.34\ttm & $0.9998$ & 8.62\ttm\\
            & & $10^{-3}$ & $0.9980$ & 6.66\ttm & $0.9977$ & 9.55\ttm & $0.9981$ & 6.56\ttm\\
            & & $10^{-2}$ & $0.9896$ & 7.64\ttm & $0.9839$ & 54.12\ttm & $0.9709$ & 4.45\ttm \\
            \cline{2-9}
            & \multirow{3}{*}{Depolarizing} & $10^{-4}$ & $0.9998$ & 6.10\ttm & $0.9998$ & 6.89\ttm & $0.9998$ & 6.77\ttm \\
            & & $10^{-3}$ & $0.9981$ & 4.51\ttm & $0.9985$ & 5.34\ttm & $0.9978$ & 21.75\ttm \\
            & & $10^{-2}$ & $0.9809$ & 1.20\tth & $0.9767$ & 6.48\ttm & $0.9773$ & 8.48\ttm\\
            \cline{2-9}
            & \multirow{3}{*}{Bit flip} & $10^{-4}$ & $0.9998$ & 6.52\ttm & $0.9999$ & 5.39\ttm & $0.9999$ & 6.86\ttm\\
            & & $10^{-3}$ & $0.9986$ & 4.38\ttm & $0.9984$ & 7.96\ttm & $0.9971$ & 10.37\ttm \\
            & & $10^{-2}$ & $0.9917$ & 5.03\tth & $0.9894$ & 4.15\tth & $0.9854$ & 3.90\tth\\
            \cline{2-9}
            & \multirow{3}{*}{Mixed noise} & $10^{-4}$ & $0.9998$ & 6.67\ttm & $0.9998$ & 5.19\ttm & $0.9997$ & 10.39\ttm \\
            & & $10^{-3}$ & $0.9976$ & 7.06\ttm & $0.9976$ & 5.91\ttm & $0.9986$ & 6.62\ttm \\
            & & $10^{-2}$ & $ 0.9806$ & 7.70\ttm & $0.9850$ & 7.98\ttm & $0.9881$ & 6.02\tth\\
            \Xhline{0.6pt}
        \end{tabular}}
        \vspace{-0.6cm}
\end{table}

\subsubsection{Evaluation Details and Results:}
We choose the models with $25$ and $27$ qubits to run experiments.
Since the parameters are randomly sampled, for each noise with different levels of probability, we generate the model and evaluate the Lipschitz constant $K^*$ for $3$ times.
However, because a $2^{25}\times 2^{25}$ or $2^{27}\times 2^{27}$ complex matrix consumes a huge amount of memory, it is not feasible to directly use Algorithm~\ref{Algorithm:K} as the previous experiment, where we represent the $M_A$ in Algorithm~\ref{Algorithm:K} as a matrix and use the package \texttt{NumPy} to evaluate eigenvalue.
We instead use a tensor network~\cite{bridgeman2017tensor} to represent the $M_A$ and
the subroutine of evaluating eigenvalue in Algorithm~\ref{Algorithm:K} is implemented with the basic power method for eigenvalue problem~\cite{doi:10.1137/1.9780898719581} by using \texttt{TensorNetwork} package.
Although there are some packages for sparse matrix in Python that can collaborate with \texttt{TensorNetwork}, their implementation for computing eigenvalues still consumes a huge amount of memory. The evaluation results on QCNN models with randomly sampled parameters and different quantum noises are listed in Table~\ref{tab:qcnn}. These results prove that our fairness verification
algorithm is efficient and can handle $27$-qubit quantum decision models on a small server. For further exploring the scalability of our verification algorithm, we also test on 29-qubit QCNN models; Please see Appendix~\ref{sec:more_experiments} for the results. 

Last but not least, it is worth noting that in all experiments, we also obtain bias kernels by Algorithm~\ref{Algorithm:K} at the running time presented in Tables.~\ref{tab:financial} and \ref{tab:qcnn}, but as they are large-size (up to $2^{27}$-dimensional) vectors, we do not show them. 

\section{Conclusion}
In this work, we initiate the studies on algorithmic verification of fairness of quantum machine learning decision models. In particular, we showed that this verification problem can be reduced to computing the Lipschitz constant of the decision models, and then resolved the latter by introducing and estimating single measurement distinguishability. Based on these theoretical results, we developed an algorithm that can verify the $(\epsilon,\delta)$-fairness of quantum decision models and provides useful bias kernels for explaining the unfairness of the models.

An interesting topic for future research is how to improve the results presented in this paper for  
training quantum decision models with fairness guarantee. On the other hand, further investigations are required to better understand the bias kernels detected by our verification algorithm, especially through more experiments on real-world applications.  
\section*{Acknowledgments}
Ji Guan would like to thank Jiayi Chen for her linguistic assistance during the preparation of this paper.
This work was partly supported by the National Key R\&D Program of China (Grant No: 2018YFA0306701), the National Natural Science Foundation of China (Grant No: 61832015). 
\bibliographystyle{unsrt} 
\bibliography{main}

\clearpage

\appendix
\section*{Appendix}
\section{Proof of Lemma~\ref{lem:Lipschitz}}\label{sec:proof_lip}
\begin{proof}
This mainly results from~\cite[Theorem 9.1]{nielsen2010quantum}: for any $\rho,\sigma \in\dh$ 
\[\tr(|\rho-\sigma|)=\max_{\{M_i\}_i}\sum_{i}|\tr(M_i^\dagger M_i(\rho-\sigma) )|\]
where the maximization is over all quantum measurements $\{M_i\}_i$. We complete the proof by noting that $\{\sqrt{\e^\dagger(M_i^\dagger M_i)}\}_i$ is a measurement and 
\[\tr(M_i^\dagger M_i\e(\rho-\sigma) )=\tr(\e^\dagger(M_i^\dagger M_i)(\rho-\sigma) ),\]
where $\e^\dagger$ is the conjugate map of $\e$.
\hfill $\Box$
\end{proof}

\section{Proof of Theorem~\ref{Thm:fairness_verification}}\label{sec:proof_verification}
\begin{proof}
    We first prove the first claim. 
    The ``if'' direction can be derived by the  definition of $(\epsilon,\delta)$-fairness, $K^*$ and Eq.(\ref{Eq:Lipschitz}) in the paper. Then we prove ``only if'' direction. By Theorem~\ref{Thm:main} in the paper, we have that there exists a pair of mutually orthogonal pure quantum states $\ket{\psi}$ and $\ket{\phi}$ such that
   \[
      d(\a(\psi),\a(\phi))=K^* D(\psi,\phi).
    \]
    Note that $D(\psi,\phi)=1$ results from  the orthogonality of $\psi$ and $\phi$, i.e., $\braket{\psi}{\phi}=0$. Given a quantum state $\sigma$, let $$\rho_{\psi}=\epsilon\psi+(1-\epsilon)\sigma\qquad \rho_{\phi}=\epsilon\phi+(1-\epsilon)\sigma.$$
    Then $D(\rho_{\psi},\rho_\phi)=\epsilon$. Subsequently,  
    \[d(\a(\rho_\psi),\a(\rho_\phi))=\epsilon d(\a(\psi),\a(\phi))=\epsilon K^*D(\psi,\phi)=\epsilon K^*\leq \delta.\]
    The above inequality follows from the definition of $(\epsilon,\delta)$-fairness of $\a$.
    
    Next, we prove the second claim.
    If the $(\epsilon,\delta)$-fairness fails, i.e., $\delta <K^*\epsilon$, then for any quantum state $\sigma$, $D(\rho_\psi ,\rho_\phi)=\epsilon$ and 
\[d(\a(\rho_\psi),\a(\rho_\phi))=K^* \epsilon D(\psi,\phi)=K^* \epsilon>\delta.\]
So $(\rho_\psi,\rho_\phi)$ is a bias pair by Definition~\ref{def:bias} in the paper.
\hfill $\Box$
\end{proof}

\section{Single Quantum  Measurement Distinguishability}\label{Measurement}
In this section, we aim to prove Theorem~\ref{Thm:main} in the paper by reducing it to a problem, called \emph{single quantum  measurement distinguishability}.

\begin{table}[ht]
\centering
\caption{Quantum Distinguishibability Problems Summary}
\begin{tabular}{|c|p{70pt}<{\centering}|p{70pt}<{\centering}|p{70pt}<{\centering}|p{70pt}<{\centering}|}
\hline
\diagbox[innerwidth=1.8cm]{Scenario}{Types}&States&Super-operators& \makecell[c]{Multiple \\ Measurements} & \makecell[c]{Single \\ Measurement} \\
\hline
Prior Known&two states & \makecell[c]{two \\ super-operators} & \makecell[c]{two \\ measurements} & \makecell[c]{one \\ measurement}\\
\hline
Aims & \makecell[c]{one \\ measurement} & one protocol& one state& two states\\
\hline
\end{tabular}
\label{table:distinguishability}
\end{table}
The distinguishability of quantum data and computation models is a fundamental problem in the field of quantum
information and has been intensively studied in the last two decades. For two prior known quantum states (super-operators), \emph{quantum states (super-operators) distinguishability} mainly considers to find an optimal measurement (protocol) to determine which one is given  through (a sequence of) measurement  outcomes~\cite{nielsen2010quantum,rudolph2003unambiguous,duan2007entanglement,duan2009perfect}. Besides this, \emph{multiple quantum measurements distinguishability}~\cite{ji2006identification,datta2021perfect} was also studied to find an optimal quantum state to distinguish two prior known quantum measurements in a similar scenario. See Table~\ref{table:distinguishability} for the summary.
A  fundamental step of these distinguishability problems is to compute the maximum probability $\pr(\{M_i\}_{i\in\o}|\rho,\sigma)$ of discriminating two quantum states $\rho$ and $\sigma$  by  measurement $\{M_i\}_{i\in\o}$. More specifically, suppose that Alice is randomly given a state $X$ from two prior known quantum states $\{\rho,\sigma\}$, and she wishes to use a quantum measurement $\{M_{i}\}_{i\in \o}$ to guess whether $X$ is state $\rho$ or state $\sigma$. 
The best strategy (successful probability) is as follows:
\begin{equation*}
\begin{aligned}
	\pr(\{M_i\}_{i\in\o}|\rho,\sigma)&=\frac{1}{2}\pr(i\in A\mid X=\rho)+\frac{1}{2}\pr(i\not\in A\mid X=\sigma)\\
 		&=\frac{1}{2}\sum_{i\in A}\tr(\m_i\rho)+\frac{1}{2}\sum_{i\not\in A}\tr(\m_i\sigma)\\
 		&=\frac{1}{2}\sum_{i\in \o}\max\{\tr(\m_i\rho), \tr(\m_i\sigma)\}\\
 		&=\frac{1}{2}\sum_{i\in\o}[\frac{\tr(\m_i(\rho+\sigma))}{2}+\frac{|\tr(\m_i(\rho-\sigma))|}{2}]\\
 		&=\frac{1}{2}+\frac{1}{4}\sum_{i\in\o}|\tr(\m_i(\rho-\sigma))|
\end{aligned}	
\end{equation*} 
where $\m_i=M_i^\dagger M_i$ for all $i\in \o$, and  $A=\{i\in\o\mid \tr(\m_i\rho)\geq \tr(\m_i\sigma)\}$ records the set of outcomes that the measurement probability of $\rho$ is higher than or equal to that of $\sigma$. In the context of  quantum distinguishability, as we only consider the probabilities of measurement outcomes, in the following, we use \emph{positive operator-valued measurement (POVM)}  $\{\m_i=M_i^\dagger M_i\}_{i\in\o}$ to replace  general measurement $\{M_i\}_{i\in\o}$. A POVM is used when we do not care the post-measurement states~\cite[Section 2.2.6]{nielsen2010quantum}.

In the context of quantum states distinguishability, Alice further tries to select a different measurement to optimize the probability $\pr(\{M_i\}_{i\in\o}|\rho,\sigma)$, i.e., 
\[\max_{\{M_i\}_{i\in\o}}\pr(\{M_i\}_{i\in\o}|\rho,\sigma)=\frac{1}{2}+\frac{1}{4}\max_{\{M_i\}_{i\in\o}}\sum_{i\in\o}|\tr(\m_i(\rho-\sigma))|=\frac{1}{2}+\frac{1}{2}D(\rho,\sigma).\]
The last equation results from~\cite[Theorem 9.1]{nielsen2010quantum}. Thus, the trace distance of $\rho$ and $\sigma$ represents quantum states distinguishability.

In this paper, we treat $\rho$ and $\sigma$ as variables to maximize $\pr(\{M_i\}_{i\in\o}|\rho,\sigma)$ which is the problem of \emph{single quantum  measurement distinguishability} (See the last column of Table~\ref{table:distinguishability}). More specifically, Alice seeks to use measurement $\{M_{i}\}_{i\in \o}$ to maximize the probability that her guess is correct by selecting  different candidate pair of quantum states $(\rho,\sigma)$. That is 
\[\max_{\rho,\sigma\in\dh}\pr(\{M_i\}_{i\in\o}|\rho,\sigma)=\frac{1}{2}+\frac{1}{4}\max_{\rho,\sigma\in\dh}\sum_{i\in\o}|\tr(\m_i(\rho-\sigma))|.\] 
As we can see, $\frac{1}{2}\max_{\rho,\sigma\in\dh}\sum_{i\in\o}|\tr(\m_i(\rho-\sigma))|$ characterizes single quantum measurement distingishability. Interestingly, it is the same to the Lipschitz constant $K^*$ of quantum decision model $\a=(\e,\{M_i\}_{i\in\o})$ when $\e=\id_\h$,  the identity super-operator on $\h$: $\id_\h(\rho)=\rho$ for all $\rho\in\dh$.



\begin{lemma}\label{lem:relationship}
Let $\a=(\e,\{M_i\}_{i\in\o})$ be a quantum decision model and $K^*$ the Lipschitz constant of $\a$. Then 
	\[2K^*=\max_{\rho,\sigma\in\dh}\sum_{i\in \o}|\tr(\e^\dagger(\m_i)(\rho-\sigma))|=\max_{\rho\perp\sigma}\sum_{i\in\o}|\tr(\e^\dagger(\m_i)(\rho-\sigma))|,\]
	where $\m_i=M_i^\dagger M_i$ for all $i\in\o$ and $\rho\perp\sigma$ means that  $\rho$ is orthogonal to $\sigma$, i.e., $\tr(\rho\sigma)=0$.
\end{lemma}

\begin{proof}
First of all, by Eq.~(\ref{Eq:Lipschitz}) in the paper,
\[K^*=\max_{\rho,\sigma\in \dh}\frac{d(\a(\rho),\a(\sigma))}{D(\rho,\sigma)}= \max_{\rho,\sigma\in \dh}\frac{\sum_{i\in\o}|\tr(\e^\dagger(\m_i)(\rho-\sigma))|}{\tr(|\rho-\sigma|).}\]

	For any $\rho,\sigma \in\dh$, let $\rho-\sigma=\Delta_{+}-\Delta_{-}$ be a decomposition into orthogonal positive and negative parts (i.e., $\Delta_{\pm}\geq 0$ and $\Delta_{+}\Delta_{-}=0$). Then we have $\tr(\Delta_+)=\tr(\Delta_{-})$ and $\tr(|\rho-\sigma|)=2\tr(\Delta_{+})$. Subsequently, 
	\begin{equation*}
		\begin{aligned}
			\frac{\sum_{i\in\o}|\tr(\e^\dagger(\m_i)(\rho-\sigma))|}{\tr(|\rho-\sigma|)}&=\frac{\sum_{i\in\o}|\tr(\e^\dagger(\m_i)(\Delta_{+}-\Delta_{-}))|}{2\tr(\Delta_{+})}\\
			&=\frac{1}{2}\sum_{i\in\o}\mid\tr(\e^\dagger(\m_i)\frac{\Delta_{+}}{\tr(\Delta_{+})}-\e^\dagger(\m_i)\frac{\Delta_{-}}{\tr(\Delta_{-})})\mid.
		\end{aligned}
	\end{equation*}
	As $\frac{\Delta_{\pm}}{\tr(\Delta_{\pm})}\in \dh$ are two orthogonal quantum states, we can obtain that
		\[2K^*=\max_{\rho\perp\sigma}\sum_{i\in\o}|\tr(\e^\dagger(\m_i)(\rho-\sigma))|.\]

		Next, we prove 
		\[\max_{\rho,\sigma\in\dh}\sum_{i\in\o}|\tr(\e^\dagger(\m_i)(\rho-\sigma))|=\max_{\rho\perp\sigma}\sum_{i\in\o}|\tr(\e^\dagger(\m_i)(\rho-\sigma))|.\]
		The r.h.s is certainly a lower bound of the l.h.s since any two orthogonal quantum  states are in $\dh$. So we have to show that it is also an upper bound. Again, for any $\rho,\sigma \in\dh$, let $\rho-\sigma=\Delta_{+}-\Delta_{-}$ be a decomposition into orthogonal positive and negative parts. Then we have  
		\[\sum_{i\in\o}|\tr(\e^\dagger(\m_i)(\rho-\sigma))|=\tr(\Delta_+)\sum_{i\in\o}\mid\tr(\e^\dagger(\m_i)\frac{\Delta_{+}}{\tr(\Delta_{+})}-\e^\dagger(\m_i)\frac{\Delta_{-}}{\tr(\Delta_{-})})\mid.\]
		As trace distance is in $[0,1]$ and  $D(\rho,\sigma)=\tr(\Delta_+)$, $\tr(\Delta_+)\leq 1$. Therefore,
		\[\sum_{i\in\o}|\tr(\e^\dagger(\m_i)(\rho-\sigma))|\leq \sum_{i\in\o}\mid\tr(\e^\dagger(\m_i)\frac{\Delta_{+}}{\tr(\Delta_{+})}-\e^\dagger(\m_i)\frac{\Delta_{-}}{\tr(\Delta_{-})})\mid.\]
		We complete the proof by noting that  $\frac{\Delta_{\pm}}{\tr(\Delta_{\pm})}\in \dh$ are two orthogonal quantum states.
		\hfill $\Box$
\end{proof}

By the above lemma, we can see that the distinguishability of a POVM  $\{\e^\dagger(\m_i)\}_{i\in \o}$ is $\frac{1}{2}(1+K^*)$. Furthermore, it is worth noting that  $\{\e^\dagger(\m_i)\}$ can be any POVM as $\e$ can be the identity super-operator $\id_\h$ on $\h$. Thus $K^*$ represents the distinguishability  of quantum POVM  and  mathematically can be redefined as a mapping:
\[K^*:\mm\rightarrow[0,1] \qquad K^*(\{\m_i\}_{i\in\o})=\max_{\rho,\sigma\in\dh}\frac{1}{2}\sum_{i\in\o}|\tr(\m_i(\rho-\sigma))|\]
where $\mm$ denotes the set of all POVMs on $\h$.

Then, with the observation that $K^*(\{\m_i\}_{i\in\o})=K^*(\{U\m_iU^
\dagger\}_{i\in\o})$ for all unitary matrix $U$, we can restate Theorem~\ref{lem:noise_increase_fairness} in the paper in the context of single quantum measurement distinguishability as follows.
 \begin{theorem}\label{lem:composition}
 For any POVM $\{\m_{i}\}_{i\in\o}$ and super-operator $\e$, we have 
 \[K^*(\{\m_{i}\}_{i\in\o})\geq K^*(\{\e^
 \dagger(\m_{i})\}_{i\in\o}).\]
\end{theorem}
 \begin{proof} By definition of $K^*(\cdot)$, we have 
 \begin{equation*}
		\begin{aligned}
			K^*(\{\e^
 \dagger(\m_{i})\}_{i\in\o})&=\max_{\rho,\sigma\in\dh}\frac{1}{2}\sum_{i\in\o}|\tr(\e^\dagger(\m_i)(\rho-\sigma))|\\
 &=\max_{\rho,\sigma\in\dh}\frac{1}{2}\sum_{i\in\o}|\tr(\m_i(\e(\rho)-\e(\sigma)))|\\
 &\leq \max_{\rho,\sigma\in\dh}\frac{1}{2}\sum_{i\in\o}|\tr(\m_i(\rho-\sigma))|\\
 &=K^*(\{\m_{i}\}_{i\in\o}).
		\end{aligned}
	\end{equation*}
The above inequality comes from $\{\e(\rho):\rho\in\dh\}\subseteq \dh$.
 \hfill $\Box$
 \end{proof}

In the following, for completing the proof of Theorem~\ref{Thm:main} in the paper, we show how to compute $K^*(\{\m_i\}_{i\in\o})$ for any POVM $\{\m_i\}_{i\in\o}$.

First, we observe that the optimization problem of computing $K^*(\{\m_i\}_{i\in\o})$  in Lemma~\ref{lem:relationship} can be constrained in pure states instead of searching mixed states.

\begin{lemma}\label{lem:pure_state}
 	\[2K^*(\{\m_i\}_{i\in\o})=\max_{\ket{\psi}\perp\ket{\phi}}\sum_{i\in\o}|\tr(\m_i(\psi-\phi))|,\]
 	where  $\ket{\psi}\perp\ket{\phi}$ means that  $\ket{\psi}$ and $\ket{\phi}$ are mutually orthogonal, i.e., $\braket{\psi}{\phi}=0$, and  $\psi=\ketbra{\psi}{\psi}$ and $
 \phi=\ketbra{\phi}{\phi}$.
 \end{lemma} 
 \begin{proof}
 	The r.h.s is certainly a lower bound of the  l.h.s since the set of pure states is a subset of $\dh$. So we have to show that it is also an upper bound. To this end, for any orthogonal $\rho\perp\sigma $, there is a probability distribution  $\{p_j\}_j$ and two sets of pure states $\{\ket{\phi_j}\}_{j}$ and $\{\ket{\psi_j}_j\}$ such that $\rho-\sigma$ has a convex decomposition:
 	\[\rho-\sigma =\sum_{j}p_j(\psi_j-\phi_j).\]
 	Then by the triangle inequality, we have  
 	\[\sum_{i\in\o}|\tr(\m_i\rho-\m_i\sigma)|\leq \max_{j} \sum_{i\in\o}|\tr(\m_i(\psi_j-\phi_j))|.\]
 	Furthermore, since for any $j$, $\ket{\psi_j}\perp\ket{\phi_j}$, we can claim that 
 	\[\max_{j} \sum_{i\in\o}|\tr(\m_i(\psi_j-\phi_j))|\leq \max_{\ket{\psi}\perp\ket{\phi}}\sum_{i\in\o}|\tr(\m_i(\psi-\phi))|.\]
\hfill $\Box$
 \end{proof}
 
 \begin{theorem}\label{thm:K*}
 	\[K^*(\{\m_i\}_{i\in\o})=\max_{A\subseteq \o}[\lambda_{\max}(\m_A)-\lambda_{\min}(\m_A)]\]
	where $\m_A=\sum_{i\in A}\m_i,$ and $\lambda_{\max}(\m_A)$ and $\lambda_{\min}(\m_A)$ are the maximum and minimum eigenvalues of positive semi-definite matrix $\m_A$, respectively.
	
	Furthermore, let $A^*\subseteq \o$ be an optimal solution of reaching $K^*(\{\m_i\}_{i\in\o})$, i.e.,
	  \[A^*=\arg\max_{A\subseteq \o}[\lambda_{\max}(\m_A)-\lambda_{\min}(\m_A)]\]
	  and $\ket{\psi}$ and $\ket{\phi}$  be two normalized eigenvectors corresponding to the maximum and minimum eigenvalues of $M_{A^*}$, respectively. Then we have 
      	\[2K^*(\{\m_i\}_{i\in\o})=\sum_{i\in\o}|\tr(\m_i(\psi-\phi))|,\]
      where $\psi=\ketbra{\psi}{\psi}$ and $\phi=\ketbra{\phi}{\phi}$.
 \end{theorem}
 \begin{proof}
By Lemma~\ref{lem:pure_state}, we have to show  	
\[\max_{\ket{\psi}\perp\ket{\phi}}\sum_{i\in\o}|\tr(\m_i(\psi-\phi))|=2\max_{A\subseteq \o}[\lambda_{\max}(\m_A)-\lambda_{\min}(\m_A)].\]
First, we claim that the l.h.s is an upper bound of the r.h.s. For any subset $A\subseteq \o$, 
\begin{equation*}
\begin{aligned}
	&\max_{\ket{\psi}\perp\ket{\phi}}\sum_{i\in\o}|\tr(\m_i(\psi-\phi))|\\
	\geq& \max_{\ket{\psi}\perp\ket{\phi}}[|\tr(\m_A(\psi-\phi))|+|\tr(\m_{\o\setminus A}(\psi-\phi))|]\\
	=&2\max_{\ket{\psi}\perp\ket{\phi}}|\tr(\m_A(\psi-\phi)|\\
	=&2\max_{\ket{\psi}\perp\ket{\phi}}\bra{\psi}\m_A\ket{\psi}-\bra{\phi}\m_A\ket{\phi}\\
	=&2[\lambda_{\max}(\m_A)-\lambda_{\min}(\m_A)]. 
\end{aligned}
\end{equation*}
The above inequality results from the triangle inequality, the first equality is obtained by $\m_A=I-\m_{\o\setminus A}$, and the last equality follows the linear algebra fact: for any positive semi-definite matrix $\m$, 
\[\max_{\ket{\psi}}\bra{\psi}\m\ket{\psi}=\lambda_{\max}(\m)\qquad \min_{\ket{\phi}}\bra{\phi}\m\ket{\phi}=\lambda_{\min}(\m)\] 
and the optimal $\ket{\psi^*}$ and $\ket{\phi^*}$ are  corresponding to maximum and minimum normalized eigenvectors, respectively; furthermore,  if $\lambda_{\min}(\m)\not=\lambda_{\max}(\m)$, then $\ket{\psi^*}\perp\ket{\phi^*}$, otherwise $\m=k I$ for some constant $k>0$ and  any two orthogonal pure states in $\h$ are the optimal $\ket{\psi^*}$ and $\ket{\phi^*}$.

Next, we prove the l.h.s is also a lower bound. Let $A=\{i\in\o\mid\tr(\m_i\psi)\geq \tr(\m_i\phi)\}$. For any $\ket{\psi}\perp\ket{\phi}$, 
\begin{equation*}
\begin{aligned}
	&\sum_{i}|\tr(\m_i(\psi-\phi))|\\
	=& \tr(\m_A(\psi-\phi))-\tr(\m_{\o\setminus A}(\psi-\phi))\\
	=&2\tr(\m_A(\psi-\phi))\\
	\leq&2[\lambda_{\max}(\m_A) -\lambda_{\min}(\m_A)]\\
	\leq&2\max_{A\subseteq \o}[\lambda_{\max}(\m_A) -\lambda_{\min}(\m_A)].
\end{aligned}
\end{equation*}
The first equality comes from the definition of set $A$, the second one is obtained by $\m_A=I-\m_{\o\setminus A}$, and the first inequality is observed by the above linear algebra fact. 
\hfill $\Box$
 \end{proof}
 
 Now we are ready to prove Theorem~\ref{Thm:main} in the paper.
\begin{proof}
 This directly results from Theorem~\ref{thm:K*} for $K^*(\{\e^{\dagger}(M_i^\dagger M_i)\}_{i\in\o})$ with POVM $\{\e^{\dagger}(M_i^\dagger M_i)\}_{i\in\o}$.
 \hfill $\Box$
\end{proof}

\section{More Scalability Experiments}\label{sec:more_experiments}
Table~\ref{tab:qcnn2} summarizes the experimental
results for computing the Lipshitz constant $K^*$ on 29-qubit QCNN models by Algorithm~\ref{Algorithm:K} in the paper.  In this experiment, the time-out (TO) is set as 10 hours.

\begin{table}[ht]
\vspace{-1cm}
    \centering
    \caption{Experimental results of Lipschitz constant $K^*$ of QCNN models.}
    \label{tab:qcnn2}
    \begin{tabular}{p{36pt}<{\centering}ccp{36pt}<{\centering}p{28pt}<{\raggedleft}|p{36pt}<{\centering}p{28pt}<{\raggedleft}|p{36pt}<{\centering}p{28pt}<{\raggedleft}}
            \Xhline{1pt}
            \multirow{2}{*}{\#Qubits} & \multicolumn{2}{c}{Noise} & \multicolumn{2}{c|}{Evaluation I} & \multicolumn{2}{c|}{Evaluation II} & \multicolumn{2}{c}{Evaluation III}\\
            \cline{2-9}
            & type & probability & {$K^*$} & \multicolumn{1}{c|}{Time} & {$K^*$} & \multicolumn{1}{c|}{Time} & {$K^*$} & \multicolumn{1}{c}{Time}\\
            \Xhline{0.6pt}
            \multirow{13}{*}{$29$} & \multicolumn{2}{c}{None} & $1.0000$ & \multicolumn{1}{c}{$\backslash$} & $1.0000$ & \multicolumn{1}{c}{$\backslash$}& $1.0000$ & \multicolumn{1}{c}{$\backslash$}\\
            \cline{2-9}
            & \multirow{3}{*}{Phase flip} & $10^{-4}$ & $0.9999$ & 25.46\ttm & $0.9999$ & 17.43\ttm & $0.9999$ & 1.24\tth\\
            & & $10^{-3}$ & $0.9982$ & 28.38\ttm & $0.9983$ & 17.51\ttm & $0.9978$ & 57.38\ttm\\
            & & $10^{-2}$ & - & \multicolumn{1}{c|}{\texttt{~~TO}} & - & \multicolumn{1}{c|}{\texttt{~~TO}} & - & \multicolumn{1}{c}{\texttt{~~TO}} \\
            \cline{2-9}
            & \multirow{3}{*}{Depolarize} & $10^{-4}$ & $0.9998$ & 22.71\ttm & $0.9998$ & 45.08\ttm & $0.9997$ & 34.94\ttm \\
            & & $10^{-3}$ & $0.9985$ & 41.53\ttm & $0.9978$ & 33.45\ttm & $0.9978$ & 35.64\ttm \\
            & & $10^{-2}$ & - & \multicolumn{1}{c|}{\texttt{~~TO}} & $0.9773$ & 1.76\tth & $0.9810$ & 1.47\tth\\
            \cline{2-9}
            & \multirow{3}{*}{Bit flip} & $10^{-4}$ & $0.9998$ & 33.60\ttm & $0.9999$ & 54.64\ttm & $0.9998$ & 1.06\tth\\
            & & $10^{-3}$ & $0.9978$ & 55.11\ttm & $0.9974$ & 1.33\tth & $0.9977$ & 40.71\ttm \\
            & & $10^{-2}$ & - & \multicolumn{1}{c|}{\texttt{~~TO}} & - & \multicolumn{1}{c|}{\texttt{~~TO}} & - & \multicolumn{1}{c}{\texttt{~~TO}}\\
            \cline{2-9}
            & \multirow{3}{*}{Mixed noise} & $10^{-4}$ & $0.9998$ & 35.4\ttm & $0.9998$ & 24.10\ttm & $0.9998$ & 36.89\ttm \\
            & & $10^{-3}$ & $0.9986$ & 20.33\ttm & $0.9977$ & 26.83\ttm & $0.9980$ & 28.05\ttm \\
            & & $10^{-2}$ & $0.9778$ & 4.74\tth & - & \multicolumn{1}{c|}{\texttt{~~TO}} & $0.9851$ & 32.45\ttm\\
            \Xhline{0.6pt}
        \end{tabular}
\end{table}
\end{document}

%% file: classification.tikz
\begin{tikzpicture}[scale=0.97,x=1pt,y=1pt]
\filldraw[color=white] (0.000000, -8.750000) rectangle (224.000000, 79.750000);
\foreach \i/\y in {1/71,2/53.5,3/36,4/17.5,5/0}{
  \node[left] at (2,\y) {$q_{\i}$};
}
\draw[color=black] (0.000000,71.000000) -- (210.000000,71.000000);
\draw[color=black] (0.000000,53.500000) -- (210.000000,53.500000);
\draw[color=black] (0.000000,36.000000) -- (210.000000,36.000000);
\draw[color=black] (0.000000,17.500000) -- (210.000000,17.500000);
\draw[color=black] (0.000000,0.000000) -- (210.000000,0.000000);
\draw[color=black,xshift=-11] (0.000000,35.500000)
node[left,xshift=0] {\Large $\bm{\to}$}
node[left,xshift=-20] {\LARGE $\rho$}
node[left,yshift=-66.25,align=center,xshift=8,inner sep=0] {\\ Input State};
\begin{scope}[color=white]
\begin{scope}[color=white]
\begin{scope}
\draw[fill=white] (4.500000, 26.750000) +(-45.000000:3.535534pt and 0.707107pt) -- +(45.000000:3.535534pt and 0.707107pt) -- +(135.000000:3.535534pt and 0.707107pt) -- +(225.000000:3.535534pt and 0.707107pt) -- cycle;
\clip (4.500000, 26.750000) +(-45.000000:3.535534pt and 0.707107pt) -- +(45.000000:3.535534pt and 0.707107pt) -- +(135.000000:3.535534pt and 0.707107pt) -- +(225.000000:3.535534pt and 0.707107pt) -- cycle;
\draw (4.500000, 26.750000) node {{}};
\end{scope}
\end{scope}
\end{scope}
\begin{scope}[color=white]
\begin{scope}[color=white]
\begin{scope}
\draw[fill=white] (13.500000, 26.750000) +(-45.000000:3.535534pt and 0.707107pt) -- +(45.000000:3.535534pt and 0.707107pt) -- +(135.000000:3.535534pt and 0.707107pt) -- +(225.000000:3.535534pt and 0.707107pt) -- cycle;
\clip (13.500000, 26.750000) +(-45.000000:3.535534pt and 0.707107pt) -- +(45.000000:3.535534pt and 0.707107pt) -- +(135.000000:3.535534pt and 0.707107pt) -- +(225.000000:3.535534pt and 0.707107pt) -- cycle;
\draw (13.500000, 26.750000) node {{}};
\end{scope}
\end{scope}
\end{scope}
\begin{scope}
\draw[fill=white] (31.250000, 71.000000) +(-45.000000:15.909903pt and 15.909903pt) -- +(45.000000:15.909903pt and 15.909903pt) -- +(135.000000:15.909903pt and 15.909903pt) -- +(225.000000:15.909903pt and 15.909903pt) -- cycle;
\clip (31.250000, 71.000000) +(-45.000000:15.909903pt and 15.909903pt) -- +(45.000000:15.909903pt and 15.909903pt) -- +(135.000000:15.909903pt and 15.909903pt) -- +(225.000000:15.909903pt and 15.909903pt) -- cycle;
\draw (31.250000, 71.000000) node {$\scriptstyle U_{1,\bm{\theta}_1}$};
\end{scope}
\begin{scope}
\draw[fill=white] (31.250000, 36.000000) +(-45.000000:15.909903pt and 15.909903pt) -- +(45.000000:15.909903pt and 15.909903pt) -- +(135.000000:15.909903pt and 15.909903pt) -- +(225.000000:15.909903pt and 15.909903pt) -- cycle;
\clip (31.250000, 36.000000) +(-45.000000:15.909903pt and 15.909903pt) -- +(45.000000:15.909903pt and 15.909903pt) -- +(135.000000:15.909903pt and 15.909903pt) -- +(225.000000:15.909903pt and 15.909903pt) -- cycle;
\draw (31.250000, 36.000000) node {$\scriptstyle U_{2,\bm{\theta}_2}$};
\end{scope}
\begin{scope}
\draw[fill=white] (31.250000, -0.000000) +(-45.000000:15.909903pt and 15.909903pt) -- +(45.000000:15.909903pt and 15.909903pt) -- +(135.000000:15.909903pt and 15.909903pt) -- +(225.000000:15.909903pt and 15.909903pt) -- cycle;
\clip (31.250000, -0.000000) +(-45.000000:15.909903pt and 15.909903pt) -- +(45.000000:15.909903pt and 15.909903pt) -- +(135.000000:15.909903pt and 15.909903pt) -- +(225.000000:15.909903pt and 15.909903pt) -- cycle;
\draw (31.250000, -0.000000) node {$\scriptstyle U_{3,\bm{\theta}_3}$};
\end{scope}
\begin{scope}[dashed]
\begin{scope}
\draw[fill=white] (52.500000, 71.000000) +(-45.000000:8.485281pt and 15.909903pt) -- +(45.000000:8.485281pt and 15.909903pt) -- +(135.000000:8.485281pt and 15.909903pt) -- +(225.000000:8.485281pt and 15.909903pt) -- cycle;
\clip (52.500000, 71.000000) +(-45.000000:8.485281pt and 15.909903pt) -- +(45.000000:8.485281pt and 15.909903pt) -- +(135.000000:8.485281pt and 15.909903pt) -- +(225.000000:8.485281pt and 15.909903pt) -- cycle;
\draw (52.500000, 71.000000) node {$\scriptstyle\e_1$};
\end{scope}
\end{scope}
\begin{scope}[dashed]
\begin{scope}
\draw[fill=white] (52.500000, 36.000000) +(-45.000000:8.485281pt and 15.909903pt) -- +(45.000000:8.485281pt and 15.909903pt) -- +(135.000000:8.485281pt and 15.909903pt) -- +(225.000000:8.485281pt and 15.909903pt) -- cycle;
\clip (52.500000, 36.000000) +(-45.000000:8.485281pt and 15.909903pt) -- +(45.000000:8.485281pt and 15.909903pt) -- +(135.000000:8.485281pt and 15.909903pt) -- +(225.000000:8.485281pt and 15.909903pt) -- cycle;
\draw (52.500000, 36.000000) node {$\scriptstyle\e_2$};
\end{scope}
\end{scope}
\begin{scope}[dashed]
\begin{scope}
\draw[fill=white] (52.500000, -0.000000) +(-45.000000:8.485281pt and 15.909903pt) -- +(45.000000:8.485281pt and 15.909903pt) -- +(135.000000:8.485281pt and 15.909903pt) -- +(225.000000:8.485281pt and 15.909903pt) -- cycle;
\clip (52.500000, -0.000000) +(-45.000000:8.485281pt and 15.909903pt) -- +(45.000000:8.485281pt and 15.909903pt) -- +(135.000000:8.485281pt and 15.909903pt) -- +(225.000000:8.485281pt and 15.909903pt) -- cycle;
\draw (52.500000, -0.000000) node {$\scriptstyle\e_3$};
\end{scope}
\end{scope}
\begin{scope}
\draw[fill=white] (73.750000, 53.500000) +(-45.000000:15.909903pt and 15.909903pt) -- +(45.000000:15.909903pt and 15.909903pt) -- +(135.000000:15.909903pt and 15.909903pt) -- +(225.000000:15.909903pt and 15.909903pt) -- cycle;
\clip (73.750000, 53.500000) +(-45.000000:15.909903pt and 15.909903pt) -- +(45.000000:15.909903pt and 15.909903pt) -- +(135.000000:15.909903pt and 15.909903pt) -- +(225.000000:15.909903pt and 15.909903pt) -- cycle;
\draw (73.750000, 53.500000) node {$\scriptstyle U_{4,\bm{\theta}_4}$};
\end{scope}
\draw (73.750000,17.500000) -- (73.750000,0.000000);
\begin{scope}
\draw[fill=white] (73.750000, 8.750000) +(-45.000000:15.909903pt and 28.284271pt) -- +(45.000000:15.909903pt and 28.284271pt) -- +(135.000000:15.909903pt and 28.284271pt) -- +(225.000000:15.909903pt and 28.284271pt) -- cycle;
\clip (73.750000, 8.750000) +(-45.000000:15.909903pt and 28.284271pt) -- +(45.000000:15.909903pt and 28.284271pt) -- +(135.000000:15.909903pt and 28.284271pt) -- +(225.000000:15.909903pt and 28.284271pt) -- cycle;
\draw (73.750000, 8.750000) node {$\scriptstyle U_{5,\bm{\theta}_5}$};
\end{scope}
\begin{scope}[dashed]
\begin{scope}
\draw[fill=white] (95.000000, 53.500000) +(-45.000000:8.485281pt and 15.909903pt) -- +(45.000000:8.485281pt and 15.909903pt) -- +(135.000000:8.485281pt and 15.909903pt) -- +(225.000000:8.485281pt and 15.909903pt) -- cycle;
\clip (95.000000, 53.500000) +(-45.000000:8.485281pt and 15.909903pt) -- +(45.000000:8.485281pt and 15.909903pt) -- +(135.000000:8.485281pt and 15.909903pt) -- +(225.000000:8.485281pt and 15.909903pt) -- cycle;
\draw (95.000000, 53.500000) node {$\scriptstyle\e_4$};
\end{scope}
\end{scope}
\draw[dashed] (95.000000,17.500000) -- (95.000000,0.000000);
\begin{scope}[dashed]
\begin{scope}
\draw[fill=white] (95.000000, 8.750000) +(-45.000000:8.485281pt and 28.284271pt) -- +(45.000000:8.485281pt and 28.284271pt) -- +(135.000000:8.485281pt and 28.284271pt) -- +(225.000000:8.485281pt and 28.284271pt) -- cycle;
\clip (95.000000, 8.750000) +(-45.000000:8.485281pt and 28.284271pt) -- +(45.000000:8.485281pt and 28.284271pt) -- +(135.000000:8.485281pt and 28.284271pt) -- +(225.000000:8.485281pt and 28.284271pt) -- cycle;
\draw (95.000000, 8.750000) node {$\scriptstyle\e_5$};
\end{scope}
\end{scope}
\begin{scope}[draw=white]
\begin{scope}
\draw[fill=white] (116.250000, 71.000000) +(-45.000000:15.909903pt and 15.909903pt) -- +(45.000000:15.909903pt and 15.909903pt) -- +(135.000000:15.909903pt and 15.909903pt) -- +(225.000000:15.909903pt and 15.909903pt) -- cycle;
\clip (116.250000, 71.000000) +(-45.000000:15.909903pt and 15.909903pt) -- +(45.000000:15.909903pt and 15.909903pt) -- +(135.000000:15.909903pt and 15.909903pt) -- +(225.000000:15.909903pt and 15.909903pt) -- cycle;
\draw (116.250000, 71.000000) node {$\cdots$};
\end{scope}
\end{scope}
\begin{scope}[draw=white]
\begin{scope}
\draw[fill=white] (116.250000, 53.500000) +(-45.000000:15.909903pt and 15.909903pt) -- +(45.000000:15.909903pt and 15.909903pt) -- +(135.000000:15.909903pt and 15.909903pt) -- +(225.000000:15.909903pt and 15.909903pt) -- cycle;
\clip (116.250000, 53.500000) +(-45.000000:15.909903pt and 15.909903pt) -- +(45.000000:15.909903pt and 15.909903pt) -- +(135.000000:15.909903pt and 15.909903pt) -- +(225.000000:15.909903pt and 15.909903pt) -- cycle;
\draw (116.250000, 53.500000) node {$\cdots$};
\end{scope}
\end{scope}
\begin{scope}[draw=white]
\begin{scope}
\draw[fill=white] (116.250000, 36.000000) +(-45.000000:15.909903pt and 15.909903pt) -- +(45.000000:15.909903pt and 15.909903pt) -- +(135.000000:15.909903pt and 15.909903pt) -- +(225.000000:15.909903pt and 15.909903pt) -- cycle;
\clip (116.250000, 36.000000) +(-45.000000:15.909903pt and 15.909903pt) -- +(45.000000:15.909903pt and 15.909903pt) -- +(135.000000:15.909903pt and 15.909903pt) -- +(225.000000:15.909903pt and 15.909903pt) -- cycle;
\draw (116.250000, 36.000000) node {$\cdots$};
\end{scope}
\end{scope}
\begin{scope}[draw=white]
\begin{scope}
\draw[fill=white] (116.250000, 17.500000) +(-45.000000:15.909903pt and 15.909903pt) -- +(45.000000:15.909903pt and 15.909903pt) -- +(135.000000:15.909903pt and 15.909903pt) -- +(225.000000:15.909903pt and 15.909903pt) -- cycle;
\clip (116.250000, 17.500000) +(-45.000000:15.909903pt and 15.909903pt) -- +(45.000000:15.909903pt and 15.909903pt) -- +(135.000000:15.909903pt and 15.909903pt) -- +(225.000000:15.909903pt and 15.909903pt) -- cycle;
\draw (116.250000, 17.500000) node {$\cdots$};
\end{scope}
\end{scope}
\begin{scope}[draw=white]
\begin{scope}
\draw[fill=white] (116.250000, -0.000000) +(-45.000000:15.909903pt and 15.909903pt) -- +(45.000000:15.909903pt and 15.909903pt) -- +(135.000000:15.909903pt and 15.909903pt) -- +(225.000000:15.909903pt and 15.909903pt) -- cycle;
\clip (116.250000, -0.000000) +(-45.000000:15.909903pt and 15.909903pt) -- +(45.000000:15.909903pt and 15.909903pt) -- +(135.000000:15.909903pt and 15.909903pt) -- +(225.000000:15.909903pt and 15.909903pt) -- cycle;
\draw (116.250000, -0.000000) node {$\cdots$};
\end{scope}
\end{scope}
\draw (142.750000,53.500000) -- (142.750000,36.000000);
\begin{scope}
\draw[fill=white] (142.750000, 44.750000) +(-45.000000:15.909903pt and 28.284271pt) -- +(45.000000:15.909903pt and 28.284271pt) -- +(135.000000:15.909903pt and 28.284271pt) -- +(225.000000:15.909903pt and 28.284271pt) -- cycle;
\clip (142.750000, 44.750000) +(-45.000000:15.909903pt and 28.284271pt) -- +(45.000000:15.909903pt and 28.284271pt) -- +(135.000000:15.909903pt and 28.284271pt) -- +(225.000000:15.909903pt and 28.284271pt) -- cycle;
\draw (142.750000, 44.750000) node {$\scriptstyle U_{L,\bm{\theta}_L}$};
\end{scope}
\draw[dashed] (164.000000,53.500000) -- (164.000000,36.000000);
\begin{scope}[dashed]
\begin{scope}
\draw[fill=white] (164.000000, 44.750000) +(-45.000000:8.485281pt and 28.284271pt) -- +(45.000000:8.485281pt and 28.284271pt) -- +(135.000000:8.485281pt and 28.284271pt) -- +(225.000000:8.485281pt and 28.284271pt) -- cycle;
\clip (164.000000, 44.750000) +(-45.000000:8.485281pt and 28.284271pt) -- +(45.000000:8.485281pt and 28.284271pt) -- +(135.000000:8.485281pt and 28.284271pt) -- +(225.000000:8.485281pt and 28.284271pt) -- cycle;
\draw (164.000000, 44.750000) node {$\scriptstyle\e_L$};
\end{scope}
\end{scope}
\begin{scope}[color=white]
\begin{scope}[color=white]
\begin{scope}
\draw[fill=white] (176.500000, 26.750000) +(-45.000000:3.535534pt and 0.707107pt) -- +(45.000000:3.535534pt and 0.707107pt) -- +(135.000000:3.535534pt and 0.707107pt) -- +(225.000000:3.535534pt and 0.707107pt) -- cycle;
\clip (176.500000, 26.750000) +(-45.000000:3.535534pt and 0.707107pt) -- +(45.000000:3.535534pt and 0.707107pt) -- +(135.000000:3.535534pt and 0.707107pt) -- +(225.000000:3.535534pt and 0.707107pt) -- cycle;
\draw (176.500000, 26.750000) node {{}};
\end{scope}
\end{scope}
\end{scope}
\begin{scope}[color=white]
\begin{scope}[color=white]
\begin{scope}
\draw[fill=white] (190.500000, 26.750000) +(-45.000000:10.606602pt and 0.707107pt) -- +(45.000000:10.606602pt and 0.707107pt) -- +(135.000000:10.606602pt and 0.707107pt) -- +(225.000000:10.606602pt and 0.707107pt) -- cycle;
\clip (190.500000, 26.750000) +(-45.000000:10.606602pt and 0.707107pt) -- +(45.000000:10.606602pt and 0.707107pt) -- +(135.000000:10.606602pt and 0.707107pt) -- +(225.000000:10.606602pt and 0.707107pt) -- cycle;
\draw (190.500000, 26.750000) node {{}};
\end{scope}
\end{scope}
\end{scope}
\draw (202.000000, -14.750000) node[text width=144pt,below,text centered] {$\{M_i\}_{i\in\mathcal{O}}$ \\ Measurement};
\draw (202.000000,71.000000) -- (202.000000,0.000000);
\begin{scope}
\draw[fill=white] (202.000000, 35.500000) +(-45.000000:14.142136pt and 66.114484pt) -- +(45.000000:14.142136pt and 66.114484pt) -- +(135.000000:14.142136pt and 66.114484pt) -- +(225.000000:14.142136pt and 66.114484pt) -- cycle;
\clip (202.000000, 35.500000) +(-45.000000:14.142136pt and 66.114484pt) -- +(45.000000:14.142136pt and 66.114484pt) -- +(135.000000:14.142136pt and 66.114484pt) -- +(225.000000:14.142136pt and 66.114484pt) -- cycle;
\begin{scope}[shift={(202.000000,35.500000)}]
\draw ([shift=(45:8)]0,-5) arc (45:135:8);\draw[-stealth] (0,-5) -- +(75:15);
\end{scope}
\end{scope}
\begin{scope}
\begin{scope}
\draw (223.500000, 35.500000) node[xshift=2] {\Large $\bm{\to}$};
\end{scope}
\end{scope}
\draw[draw opacity=1.000000,fill opacity=0.200000,color=black,dotted,thick,rounded corners=5pt] (10.000000,85.750000) rectangle (180.000000,-14.750000);
\draw (95.000000, -14.750000) node[below,align=center,color=black] {$\mathcal{E}_{\bm{\theta}}$ \\ Parameterized Quantum Circuit};
\draw (260.000000, -14.750000) node[below,align=center] {$\mathcal{A}_{\bm{\theta}}(\rho)$ \\ Output};
\begin{scope}
\begin{scope}
\draw (260.000000, 35.500000) node {\Large $\{p_i\}_{i\in\mathcal{O}}$};
\end{scope}
\end{scope}
\draw[decorate,decoration={brace,amplitude = 7.500000pt}, thick] (8.000000,87.750000) -- (215.000000,87.750000);
\draw (111.500000, 95.250000) node[above] {\Large $\mathcal{A}_{\bm{\theta}}$};
\end{tikzpicture}

%% file: rotation.tikz
\begin{tikzpicture}[scale=1.000000,x=1pt,y=1pt]
\filldraw[color=white] (0.000000, -6.000000) rectangle (258.000000, 46.000000);
\draw[color=black] (0.000000,40.000000) -- (250.000000,40.000000);
\draw[color=black] (0.000000,28.000000) -- (250.000000,28.000000);
\draw[color=black] (0.000000,12.000000) -- (250.000000,12.000000);
\draw[color=black] (0.000000,0.000000) -- (250.000000,0.000000);
\begin{scope}[color=white]
\begin{scope}[color=white]
\begin{scope}
\draw[fill=white] (2.500000, 20.000000) +(-45.000000:0.707107pt and 6.363961pt) -- +(45.000000:0.707107pt and 6.363961pt) -- +(135.000000:0.707107pt and 6.363961pt) -- +(225.000000:0.707107pt and 6.363961pt) -- cycle;
\clip (2.500000, 20.000000) +(-45.000000:0.707107pt and 6.363961pt) -- +(45.000000:0.707107pt and 6.363961pt) -- +(135.000000:0.707107pt and 6.363961pt) -- +(225.000000:0.707107pt and 6.363961pt) -- cycle;
\draw (2.500000, 20.000000) node {{}};
\end{scope}
\end{scope}
\end{scope}
\begin{scope}[color=white]
\begin{scope}[color=white]
\begin{scope}
\draw[fill=white] (7.000000, 20.000000) +(-45.000000:0.000000pt and 6.363961pt) -- +(45.000000:0.000000pt and 6.363961pt) -- +(135.000000:0.000000pt and 6.363961pt) -- +(225.000000:0.000000pt and 6.363961pt) -- cycle;
\clip (7.000000, 20.000000) +(-45.000000:0.000000pt and 6.363961pt) -- +(45.000000:0.000000pt and 6.363961pt) -- +(135.000000:0.000000pt and 6.363961pt) -- +(225.000000:0.000000pt and 6.363961pt) -- cycle;
\draw (7.000000, 20.000000) node {{}};
\end{scope}
\end{scope}
\end{scope}
\begin{scope}
\draw[fill=white] (15.500000, 40.000000) +(-45.000000:6.363961pt and 6.363961pt) -- +(45.000000:6.363961pt and 6.363961pt) -- +(135.000000:6.363961pt and 6.363961pt) -- +(225.000000:6.363961pt and 6.363961pt) -- cycle;
\clip (15.500000, 40.000000) +(-45.000000:6.363961pt and 6.363961pt) -- +(45.000000:6.363961pt and 6.363961pt) -- +(135.000000:6.363961pt and 6.363961pt) -- +(225.000000:6.363961pt and 6.363961pt) -- cycle;
\draw (15.500000, 40.000000) node {$Z$};
\end{scope}
\begin{scope}
\draw[fill=white] (15.500000, 28.000000) +(-45.000000:6.363961pt and 6.363961pt) -- +(45.000000:6.363961pt and 6.363961pt) -- +(135.000000:6.363961pt and 6.363961pt) -- +(225.000000:6.363961pt and 6.363961pt) -- cycle;
\clip (15.500000, 28.000000) +(-45.000000:6.363961pt and 6.363961pt) -- +(45.000000:6.363961pt and 6.363961pt) -- +(135.000000:6.363961pt and 6.363961pt) -- +(225.000000:6.363961pt and 6.363961pt) -- cycle;
\draw (15.500000, 28.000000) node {$Z$};
\end{scope}
\begin{scope}
\draw[fill=white] (15.500000, 12.000000) +(-45.000000:6.363961pt and 6.363961pt) -- +(45.000000:6.363961pt and 6.363961pt) -- +(135.000000:6.363961pt and 6.363961pt) -- +(225.000000:6.363961pt and 6.363961pt) -- cycle;
\clip (15.500000, 12.000000) +(-45.000000:6.363961pt and 6.363961pt) -- +(45.000000:6.363961pt and 6.363961pt) -- +(135.000000:6.363961pt and 6.363961pt) -- +(225.000000:6.363961pt and 6.363961pt) -- cycle;
\draw (15.500000, 12.000000) node {$Z$};
\end{scope}
\begin{scope}
\draw[fill=white] (15.500000, -0.000000) +(-45.000000:6.363961pt and 6.363961pt) -- +(45.000000:6.363961pt and 6.363961pt) -- +(135.000000:6.363961pt and 6.363961pt) -- +(225.000000:6.363961pt and 6.363961pt) -- cycle;
\clip (15.500000, -0.000000) +(-45.000000:6.363961pt and 6.363961pt) -- +(45.000000:6.363961pt and 6.363961pt) -- +(135.000000:6.363961pt and 6.363961pt) -- +(225.000000:6.363961pt and 6.363961pt) -- cycle;
\draw (15.500000, -0.000000) node {$Z$};
\end{scope}
\begin{scope}
\draw[fill=white] (28.500000, 40.000000) +(-45.000000:6.363961pt and 6.363961pt) -- +(45.000000:6.363961pt and 6.363961pt) -- +(135.000000:6.363961pt and 6.363961pt) -- +(225.000000:6.363961pt and 6.363961pt) -- cycle;
\clip (28.500000, 40.000000) +(-45.000000:6.363961pt and 6.363961pt) -- +(45.000000:6.363961pt and 6.363961pt) -- +(135.000000:6.363961pt and 6.363961pt) -- +(225.000000:6.363961pt and 6.363961pt) -- cycle;
\draw (28.500000, 40.000000) node {$X$};
\end{scope}
\begin{scope}
\draw[fill=white] (28.500000, 28.000000) +(-45.000000:6.363961pt and 6.363961pt) -- +(45.000000:6.363961pt and 6.363961pt) -- +(135.000000:6.363961pt and 6.363961pt) -- +(225.000000:6.363961pt and 6.363961pt) -- cycle;
\clip (28.500000, 28.000000) +(-45.000000:6.363961pt and 6.363961pt) -- +(45.000000:6.363961pt and 6.363961pt) -- +(135.000000:6.363961pt and 6.363961pt) -- +(225.000000:6.363961pt and 6.363961pt) -- cycle;
\draw (28.500000, 28.000000) node {$X$};
\end{scope}
\begin{scope}
\draw[fill=white] (28.500000, 12.000000) +(-45.000000:6.363961pt and 6.363961pt) -- +(45.000000:6.363961pt and 6.363961pt) -- +(135.000000:6.363961pt and 6.363961pt) -- +(225.000000:6.363961pt and 6.363961pt) -- cycle;
\clip (28.500000, 12.000000) +(-45.000000:6.363961pt and 6.363961pt) -- +(45.000000:6.363961pt and 6.363961pt) -- +(135.000000:6.363961pt and 6.363961pt) -- +(225.000000:6.363961pt and 6.363961pt) -- cycle;
\draw (28.500000, 12.000000) node {$X$};
\end{scope}
\begin{scope}
\draw[fill=white] (28.500000, -0.000000) +(-45.000000:6.363961pt and 6.363961pt) -- +(45.000000:6.363961pt and 6.363961pt) -- +(135.000000:6.363961pt and 6.363961pt) -- +(225.000000:6.363961pt and 6.363961pt) -- cycle;
\clip (28.500000, -0.000000) +(-45.000000:6.363961pt and 6.363961pt) -- +(45.000000:6.363961pt and 6.363961pt) -- +(135.000000:6.363961pt and 6.363961pt) -- +(225.000000:6.363961pt and 6.363961pt) -- cycle;
\draw (28.500000, -0.000000) node {$X$};
\end{scope}
\begin{scope}
\draw[fill=white] (41.500000, 40.000000) +(-45.000000:6.363961pt and 6.363961pt) -- +(45.000000:6.363961pt and 6.363961pt) -- +(135.000000:6.363961pt and 6.363961pt) -- +(225.000000:6.363961pt and 6.363961pt) -- cycle;
\clip (41.500000, 40.000000) +(-45.000000:6.363961pt and 6.363961pt) -- +(45.000000:6.363961pt and 6.363961pt) -- +(135.000000:6.363961pt and 6.363961pt) -- +(225.000000:6.363961pt and 6.363961pt) -- cycle;
\draw (41.500000, 40.000000) node {$Z$};
\end{scope}
\begin{scope}
\draw[fill=white] (41.500000, 28.000000) +(-45.000000:6.363961pt and 6.363961pt) -- +(45.000000:6.363961pt and 6.363961pt) -- +(135.000000:6.363961pt and 6.363961pt) -- +(225.000000:6.363961pt and 6.363961pt) -- cycle;
\clip (41.500000, 28.000000) +(-45.000000:6.363961pt and 6.363961pt) -- +(45.000000:6.363961pt and 6.363961pt) -- +(135.000000:6.363961pt and 6.363961pt) -- +(225.000000:6.363961pt and 6.363961pt) -- cycle;
\draw (41.500000, 28.000000) node {$Z$};
\end{scope}
\begin{scope}
\draw[fill=white] (41.500000, 12.000000) +(-45.000000:6.363961pt and 6.363961pt) -- +(45.000000:6.363961pt and 6.363961pt) -- +(135.000000:6.363961pt and 6.363961pt) -- +(225.000000:6.363961pt and 6.363961pt) -- cycle;
\clip (41.500000, 12.000000) +(-45.000000:6.363961pt and 6.363961pt) -- +(45.000000:6.363961pt and 6.363961pt) -- +(135.000000:6.363961pt and 6.363961pt) -- +(225.000000:6.363961pt and 6.363961pt) -- cycle;
\draw (41.500000, 12.000000) node {$Z$};
\end{scope}
\begin{scope}
\draw[fill=white] (41.500000, -0.000000) +(-45.000000:6.363961pt and 6.363961pt) -- +(45.000000:6.363961pt and 6.363961pt) -- +(135.000000:6.363961pt and 6.363961pt) -- +(225.000000:6.363961pt and 6.363961pt) -- cycle;
\clip (41.500000, -0.000000) +(-45.000000:6.363961pt and 6.363961pt) -- +(45.000000:6.363961pt and 6.363961pt) -- +(135.000000:6.363961pt and 6.363961pt) -- +(225.000000:6.363961pt and 6.363961pt) -- cycle;
\draw (41.500000, -0.000000) node {$Z$};
\end{scope}
\begin{scope}[color=white]
\begin{scope}[color=white]
\begin{scope}
\draw[fill=white] (50.000000, 20.000000) +(-45.000000:0.000000pt and 6.363961pt) -- +(45.000000:0.000000pt and 6.363961pt) -- +(135.000000:0.000000pt and 6.363961pt) -- +(225.000000:0.000000pt and 6.363961pt) -- cycle;
\clip (50.000000, 20.000000) +(-45.000000:0.000000pt and 6.363961pt) -- +(45.000000:0.000000pt and 6.363961pt) -- +(135.000000:0.000000pt and 6.363961pt) -- +(225.000000:0.000000pt and 6.363961pt) -- cycle;
\draw (50.000000, 20.000000) node {{}};
\end{scope}
\end{scope}
\end{scope}
\begin{scope}[color=white]
\begin{scope}[color=white]
\begin{scope}
\draw[fill=white] (54.000000, 20.000000) +(-45.000000:0.000000pt and 6.363961pt) -- +(45.000000:0.000000pt and 6.363961pt) -- +(135.000000:0.000000pt and 6.363961pt) -- +(225.000000:0.000000pt and 6.363961pt) -- cycle;
\clip (54.000000, 20.000000) +(-45.000000:0.000000pt and 6.363961pt) -- +(45.000000:0.000000pt and 6.363961pt) -- +(135.000000:0.000000pt and 6.363961pt) -- +(225.000000:0.000000pt and 6.363961pt) -- cycle;
\draw (54.000000, 20.000000) node {{}};
\end{scope}
\end{scope}
\end{scope}
\draw[xshift=4] (68.000000,40.000000) -- (68.000000,28.000000);
\begin{scope}[xshift=4]
\begin{scope}
\draw[fill=white] (68.000000, 40.000000) +(-45.000000:14.142136pt and 6.363961pt) -- +(45.000000:14.142136pt and 6.363961pt) -- +(135.000000:14.142136pt and 6.363961pt) -- +(225.000000:14.142136pt and 6.363961pt) -- cycle;
\clip (68.000000, 40.000000) +(-45.000000:14.142136pt and 6.363961pt) -- +(45.000000:14.142136pt and 6.363961pt) -- +(135.000000:14.142136pt and 6.363961pt) -- +(225.000000:14.142136pt and 6.363961pt) -- cycle;
\draw (68.000000, 40.000000) node {$XX$};
\end{scope}
\end{scope}
\begin{scope}[xshift=4]
\begin{scope}
\draw[fill=white] (68.000000, 28.000000) +(-45.000000:14.142136pt and 6.363961pt) -- +(45.000000:14.142136pt and 6.363961pt) -- +(135.000000:14.142136pt and 6.363961pt) -- +(225.000000:14.142136pt and 6.363961pt) -- cycle;
\clip (68.000000, 28.000000) +(-45.000000:14.142136pt and 6.363961pt) -- +(45.000000:14.142136pt and 6.363961pt) -- +(135.000000:14.142136pt and 6.363961pt) -- +(225.000000:14.142136pt and 6.363961pt) -- cycle;
\draw (68.000000, 28.000000) node {$XX$};
\end{scope}
\end{scope}
\draw[xshift=4] (92.000000,28.000000) -- (92.000000,12.000000);
\begin{scope}[xshift=4]
\begin{scope}
\draw[fill=white] (92.000000, 28.000000) +(-45.000000:14.142136pt and 6.363961pt) -- +(45.000000:14.142136pt and 6.363961pt) -- +(135.000000:14.142136pt and 6.363961pt) -- +(225.000000:14.142136pt and 6.363961pt) -- cycle;
\clip (92.000000, 28.000000) +(-45.000000:14.142136pt and 6.363961pt) -- +(45.000000:14.142136pt and 6.363961pt) -- +(135.000000:14.142136pt and 6.363961pt) -- +(225.000000:14.142136pt and 6.363961pt) -- cycle;
\draw (92.000000, 28.000000) node {$XX$};
\end{scope}
\end{scope}
\begin{scope}[xshift=4]
\begin{scope}
\draw[fill=white] (92.000000, 12.000000) +(-45.000000:14.142136pt and 6.363961pt) -- +(45.000000:14.142136pt and 6.363961pt) -- +(135.000000:14.142136pt and 6.363961pt) -- +(225.000000:14.142136pt and 6.363961pt) -- cycle;
\clip (92.000000, 12.000000) +(-45.000000:14.142136pt and 6.363961pt) -- +(45.000000:14.142136pt and 6.363961pt) -- +(135.000000:14.142136pt and 6.363961pt) -- +(225.000000:14.142136pt and 6.363961pt) -- cycle;
\draw (92.000000, 12.000000) node {$XX$};
\end{scope}
\end{scope}
\draw[xshift=4] (116.000000,12.000000) -- (116.000000,0.000000);
\begin{scope}[xshift=4]
\begin{scope}
\draw[fill=white] (116.000000, 12.000000) +(-45.000000:14.142136pt and 6.363961pt) -- +(45.000000:14.142136pt and 6.363961pt) -- +(135.000000:14.142136pt and 6.363961pt) -- +(225.000000:14.142136pt and 6.363961pt) -- cycle;
\clip (116.000000, 12.000000) +(-45.000000:14.142136pt and 6.363961pt) -- +(45.000000:14.142136pt and 6.363961pt) -- +(135.000000:14.142136pt and 6.363961pt) -- +(225.000000:14.142136pt and 6.363961pt) -- cycle;
\draw (116.000000, 12.000000) node {$XX$};
\end{scope}
\end{scope}
\begin{scope}[xshift=4]
\begin{scope}
\draw[fill=white] (116.000000, -0.000000) +(-45.000000:14.142136pt and 6.363961pt) -- +(45.000000:14.142136pt and 6.363961pt) -- +(135.000000:14.142136pt and 6.363961pt) -- +(225.000000:14.142136pt and 6.363961pt) -- cycle;
\clip (116.000000, -0.000000) +(-45.000000:14.142136pt and 6.363961pt) -- +(45.000000:14.142136pt and 6.363961pt) -- +(135.000000:14.142136pt and 6.363961pt) -- +(225.000000:14.142136pt and 6.363961pt) -- cycle;
\draw (116.000000, -0.000000) node {$XX$};
\end{scope}
\end{scope}
\draw[xshift=4] (140.000000,40.000000) -- (140.000000,0.000000);
\begin{scope}[xshift=4]
\begin{scope}
\draw[fill=white] (140.000000, -0.000000) +(-45.000000:14.142136pt and 6.363961pt) -- +(45.000000:14.142136pt and 6.363961pt) -- +(135.000000:14.142136pt and 6.363961pt) -- +(225.000000:14.142136pt and 6.363961pt) -- cycle;
\clip (140.000000, -0.000000) +(-45.000000:14.142136pt and 6.363961pt) -- +(45.000000:14.142136pt and 6.363961pt) -- +(135.000000:14.142136pt and 6.363961pt) -- +(225.000000:14.142136pt and 6.363961pt) -- cycle;
\draw (140.000000, -0.000000) node {$XX$};
\end{scope}
\end{scope}
\begin{scope}[xshift=4]
\begin{scope}
\draw[fill=white] (140.000000, 40.000000) +(-45.000000:14.142136pt and 6.363961pt) -- +(45.000000:14.142136pt and 6.363961pt) -- +(135.000000:14.142136pt and 6.363961pt) -- +(225.000000:14.142136pt and 6.363961pt) -- cycle;
\clip (140.000000, 40.000000) +(-45.000000:14.142136pt and 6.363961pt) -- +(45.000000:14.142136pt and 6.363961pt) -- +(135.000000:14.142136pt and 6.363961pt) -- +(225.000000:14.142136pt and 6.363961pt) -- cycle;
\draw (140.000000, 40.000000) node {$XX$};
\end{scope}
\end{scope}
\begin{scope}[color=white]
\begin{scope}[color=white,xshift=4]
\begin{scope}
\draw[fill=white] (154.000000, 20.000000) +(-45.000000:0.000000pt and 6.363961pt) -- +(45.000000:0.000000pt and 6.363961pt) -- +(135.000000:0.000000pt and 6.363961pt) -- +(225.000000:0.000000pt and 6.363961pt) -- cycle;
\clip (154.000000, 20.000000) +(-45.000000:0.000000pt and 6.363961pt) -- +(45.000000:0.000000pt and 6.363961pt) -- +(135.000000:0.000000pt and 6.363961pt) -- +(225.000000:0.000000pt and 6.363961pt) -- cycle;
\draw (154.000000, 20.000000) node {{}};
\end{scope}
\end{scope}
\end{scope}
\begin{scope}[color=white]
\begin{scope}[color=white,xshift=4]
\begin{scope}
\draw[fill=white] (159.500000, 20.000000) +(-45.000000:2.121320pt and 6.363961pt) -- +(45.000000:2.121320pt and 6.363961pt) -- +(135.000000:2.121320pt and 6.363961pt) -- +(225.000000:2.121320pt and 6.363961pt) -- cycle;
\clip (159.500000, 20.000000) +(-45.000000:2.121320pt and 6.363961pt) -- +(45.000000:2.121320pt and 6.363961pt) -- +(135.000000:2.121320pt and 6.363961pt) -- +(225.000000:2.121320pt and 6.363961pt) -- cycle;
\draw (159.500000, 20.000000) node {{}};
\end{scope}
\end{scope}
\end{scope}
\draw (172.500000,40.000000) -- (172.500000,0.000000);
\begin{scope}
\draw[fill=white] (172.500000, 20.000000) +(-45.000000:10.606602pt and 34.648232pt) -- +(45.000000:10.606602pt and 34.648232pt) -- +(135.000000:10.606602pt and 34.648232pt) -- +(225.000000:10.606602pt and 34.648232pt) -- cycle;
\clip (172.500000, 20.000000) +(-45.000000:10.606602pt and 34.648232pt) -- +(45.000000:10.606602pt and 34.648232pt) -- +(135.000000:10.606602pt and 34.648232pt) -- +(225.000000:10.606602pt and 34.648232pt) -- cycle;
\draw (172.500000, 20.000000) node {\rotatebox{90}{Rotation}};
\end{scope}
\begin{scope}
\begin{scope}[shift={(187.500000,20.000000)}]
\node[xshift=1.5] {$\cdots$};
\end{scope}
\end{scope}
\draw (206.500000,40.000000) -- (206.500000,0.000000);
\begin{scope}
\draw[fill=white] (206.500000, 20.000000) +(-45.000000:10.606602pt and 34.648232pt) -- +(45.000000:10.606602pt and 34.648232pt) -- +(135.000000:10.606602pt and 34.648232pt) -- +(225.000000:10.606602pt and 34.648232pt) -- cycle;
\clip (206.500000, 20.000000) +(-45.000000:10.606602pt and 34.648232pt) -- +(45.000000:10.606602pt and 34.648232pt) -- +(135.000000:10.606602pt and 34.648232pt) -- +(225.000000:10.606602pt and 34.648232pt) -- cycle;
\draw (206.500000, 20.000000) node {\rotatebox{90}{Entangling}};
\end{scope}
\draw (225.500000,40.000000) -- (225.500000,0.000000);
\begin{scope}
\draw[fill=white] (225.500000, 20.000000) +(-45.000000:10.606602pt and 34.648232pt) -- +(45.000000:10.606602pt and 34.648232pt) -- +(135.000000:10.606602pt and 34.648232pt) -- +(225.000000:10.606602pt and 34.648232pt) -- cycle;
\clip (225.500000, 20.000000) +(-45.000000:10.606602pt and 34.648232pt) -- +(45.000000:10.606602pt and 34.648232pt) -- +(135.000000:10.606602pt and 34.648232pt) -- +(225.000000:10.606602pt and 34.648232pt) -- cycle;
\draw (225.500000, 20.000000) node {\rotatebox{90}{Rotation}};
\end{scope}
\begin{scope}[color=white]
\begin{scope}[color=white]
\begin{scope}
\draw[fill=white] (239.000000, 20.000000) +(-45.000000:2.828427pt and 6.363961pt) -- +(45.000000:2.828427pt and 6.363961pt) -- +(135.000000:2.828427pt and 6.363961pt) -- +(225.000000:2.828427pt and 6.363961pt) -- cycle;
\clip (239.000000, 20.000000) +(-45.000000:2.828427pt and 6.363961pt) -- +(45.000000:2.828427pt and 6.363961pt) -- +(135.000000:2.828427pt and 6.363961pt) -- +(225.000000:2.828427pt and 6.363961pt) -- cycle;
\draw (239.000000, 20.000000) node {{}};
\end{scope}
\end{scope}
\end{scope}
\draw (248.500000, -6.000000) node[below] {$M$};
\draw (248.500000,40.000000) -- (248.500000,0.000000);
\begin{scope}
\draw[fill=white] (248.500000, 20.000000) +(-45.000000:10.606602pt and 34.648232pt) -- +(45.000000:10.606602pt and 34.648232pt) -- +(135.000000:10.606602pt and 34.648232pt) -- +(225.000000:10.606602pt and 34.648232pt) -- cycle;
\clip (248.500000, 20.000000) +(-45.000000:10.606602pt and 34.648232pt) -- +(45.000000:10.606602pt and 34.648232pt) -- +(135.000000:10.606602pt and 34.648232pt) -- +(225.000000:10.606602pt and 34.648232pt) -- cycle;
\begin{scope}[shift={(248.500000,20.000000)}]
\draw ([shift=(45:8)]0,-5) arc (45:135:8);\draw[-stealth] (0,-5) -- +(75:15);
\end{scope}
\end{scope}
\draw[draw opacity=1.000000,fill opacity=0.200000,color=black,dotted,thick] (6.000000,46.000000) rectangle (51.000000,-6.000000);
\draw (28.500000, -6.000000) node[below,color=black] {Rotation};
\draw[draw opacity=1.000000,fill opacity=0.200000,color=black,dotted,thick] (6.000000,46.000000) rectangle (51.000000,-6.000000);
\draw[draw opacity=1.000000,fill opacity=0.200000,color=black,xshift=4,dotted,thick,rounded corners=6pt] (53.000000,46.000000) rectangle (155.000000,-6.000000);
\draw (104.000000, -6.000000) node[below,color=black] {Entangling};
\draw[draw opacity=1.000000,fill opacity=0.200000,color=black,xshift=4,dotted,thick,rounded corners=6pt] (53.000000,46.000000) rectangle (155.000000,-6.000000);
\draw[decorate,decoration={brace,mirror,amplitude = 3.000000pt},thick] (164.000000,-6.000000) -- (234.000000,-6.000000);
\draw (199.000000, -9.000000) node[below] {Alternating Blocks};
\end{tikzpicture}

%% file: qcnn.tikz
\begin{tikzpicture}[scale=0.95,x=1pt,y=1pt]
\filldraw[color=white] (0.000000, -6.000000) rectangle (134.000000, 91.000000);
\draw[color=black] (0.000000,85.000000) -- (134.000000,85.000000);
\draw[color=black] (0.000000,73.000000) -- (134.000000,73.000000);
\draw[color=black] (0.000000,61.000000) -- (134.000000,61.000000);
\draw[color=black] (0.000000,48.000000) -- (134.000000,48.000000);
\draw[color=black] (0.000000,36.000000) -- (240.000000,36.000000);
\draw[color=black] (0.000000,24.000000) -- (240.000000,24.000000);
\draw[color=black] (0.000000,12.000000) -- (270.000000,12.000000);
\draw[color=black] (0.000000,0.000000) -- (285.500000,0.000000);
\draw[color=black] (285.500000,-0.500000) -- (292.000000,-0.500000);
\draw[color=black] (285.500000,0.500000) -- (292.000000,0.500000);
\begin{scope}[color=white]
\begin{scope}[color=white]
\begin{scope}
\draw[fill=white] (2.500000, 54.500000) +(-45.000000:0.707107pt and 6.363961pt) -- +(45.000000:0.707107pt and 6.363961pt) -- +(135.000000:0.707107pt and 6.363961pt) -- +(225.000000:0.707107pt and 6.363961pt) -- cycle;
\clip (2.500000, 54.500000) +(-45.000000:0.707107pt and 6.363961pt) -- +(45.000000:0.707107pt and 6.363961pt) -- +(135.000000:0.707107pt and 6.363961pt) -- +(225.000000:0.707107pt and 6.363961pt) -- cycle;
\draw (2.500000, 54.500000) node {{}};
\end{scope}
\end{scope}
\end{scope}
\begin{scope}[color=white]
\begin{scope}[color=white]
\begin{scope}
\draw[fill=white] (7.000000, 54.500000) +(-45.000000:0.000000pt and 6.363961pt) -- +(45.000000:0.000000pt and 6.363961pt) -- +(135.000000:0.000000pt and 6.363961pt) -- +(225.000000:0.000000pt and 6.363961pt) -- cycle;
\clip (7.000000, 54.500000) +(-45.000000:0.000000pt and 6.363961pt) -- +(45.000000:0.000000pt and 6.363961pt) -- +(135.000000:0.000000pt and 6.363961pt) -- +(225.000000:0.000000pt and 6.363961pt) -- cycle;
\draw (7.000000, 54.500000) node {{}};
\end{scope}
\end{scope}
\end{scope}
\draw (18.000000,85.000000) -- (18.000000,73.000000);
\begin{scope}
\draw[fill=white] (18.000000, 79.000000) +(-45.000000:9.899495pt and 14.849242pt) -- +(45.000000:9.899495pt and 14.849242pt) -- +(135.000000:9.899495pt and 14.849242pt) -- +(225.000000:9.899495pt and 14.849242pt) -- cycle;
\clip (18.000000, 79.000000) +(-45.000000:9.899495pt and 14.849242pt) -- +(45.000000:9.899495pt and 14.849242pt) -- +(135.000000:9.899495pt and 14.849242pt) -- +(225.000000:9.899495pt and 14.849242pt) -- cycle;
\draw (18.000000, 79.000000) node {$C_1$};
\end{scope}
\draw (18.000000,61.000000) -- (18.000000,48.000000);
\begin{scope}
\draw[fill=white] (18.000000, 54.500000) +(-45.000000:9.899495pt and 15.556349pt) -- +(45.000000:9.899495pt and 15.556349pt) -- +(135.000000:9.899495pt and 15.556349pt) -- +(225.000000:9.899495pt and 15.556349pt) -- cycle;
\clip (18.000000, 54.500000) +(-45.000000:9.899495pt and 15.556349pt) -- +(45.000000:9.899495pt and 15.556349pt) -- +(135.000000:9.899495pt and 15.556349pt) -- +(225.000000:9.899495pt and 15.556349pt) -- cycle;
\draw (18.000000, 54.500000) node {$C_2$};
\end{scope}
\draw (18.000000,36.000000) -- (18.000000,24.000000);
\begin{scope}
\draw[fill=white] (18.000000, 30.000000) +(-45.000000:9.899495pt and 14.849242pt) -- +(45.000000:9.899495pt and 14.849242pt) -- +(135.000000:9.899495pt and 14.849242pt) -- +(225.000000:9.899495pt and 14.849242pt) -- cycle;
\clip (18.000000, 30.000000) +(-45.000000:9.899495pt and 14.849242pt) -- +(45.000000:9.899495pt and 14.849242pt) -- +(135.000000:9.899495pt and 14.849242pt) -- +(225.000000:9.899495pt and 14.849242pt) -- cycle;
\draw (18.000000, 30.000000) node {$C_3$};
\end{scope}
\draw (18.000000,12.000000) -- (18.000000,0.000000);
\begin{scope}
\draw[fill=white] (18.000000, 6.000000) +(-45.000000:9.899495pt and 14.849242pt) -- +(45.000000:9.899495pt and 14.849242pt) -- +(135.000000:9.899495pt and 14.849242pt) -- +(225.000000:9.899495pt and 14.849242pt) -- cycle;
\clip (18.000000, 6.000000) +(-45.000000:9.899495pt and 14.849242pt) -- +(45.000000:9.899495pt and 14.849242pt) -- +(135.000000:9.899495pt and 14.849242pt) -- +(225.000000:9.899495pt and 14.849242pt) -- cycle;
\draw (18.000000, 6.000000) node {$C_4$};
\end{scope}
\draw (36.000000,73.000000) -- (36.000000,61.000000);
\begin{scope}
\draw[fill=white] (36.000000, 67.000000) +(-45.000000:9.899495pt and 14.849242pt) -- +(45.000000:9.899495pt and 14.849242pt) -- +(135.000000:9.899495pt and 14.849242pt) -- +(225.000000:9.899495pt and 14.849242pt) -- cycle;
\clip (36.000000, 67.000000) +(-45.000000:9.899495pt and 14.849242pt) -- +(45.000000:9.899495pt and 14.849242pt) -- +(135.000000:9.899495pt and 14.849242pt) -- +(225.000000:9.899495pt and 14.849242pt) -- cycle;
\draw (36.000000, 67.000000) node {$C_5$};
\end{scope}
\draw (36.000000,48.000000) -- (36.000000,36.000000);
\begin{scope}
\draw[fill=white] (36.000000, 42.000000) +(-45.000000:9.899495pt and 14.849242pt) -- +(45.000000:9.899495pt and 14.849242pt) -- +(135.000000:9.899495pt and 14.849242pt) -- +(225.000000:9.899495pt and 14.849242pt) -- cycle;
\clip (36.000000, 42.000000) +(-45.000000:9.899495pt and 14.849242pt) -- +(45.000000:9.899495pt and 14.849242pt) -- +(135.000000:9.899495pt and 14.849242pt) -- +(225.000000:9.899495pt and 14.849242pt) -- cycle;
\draw (36.000000, 42.000000) node {$C_6$};
\end{scope}
\draw (36.000000,24.000000) -- (36.000000,12.000000);
\begin{scope}
\draw[fill=white] (36.000000, 18.000000) +(-45.000000:9.899495pt and 14.849242pt) -- +(45.000000:9.899495pt and 14.849242pt) -- +(135.000000:9.899495pt and 14.849242pt) -- +(225.000000:9.899495pt and 14.849242pt) -- cycle;
\clip (36.000000, 18.000000) +(-45.000000:9.899495pt and 14.849242pt) -- +(45.000000:9.899495pt and 14.849242pt) -- +(135.000000:9.899495pt and 14.849242pt) -- +(225.000000:9.899495pt and 14.849242pt) -- cycle;
\draw (36.000000, 18.000000) node {$C_7$};
\end{scope}
\draw (50.000000,85.000000) -- (50.000000,0.000000);
\begin{scope}
\draw[fill=white] (50.000000, -0.000000) +(-45.000000:9.899495pt and 6.363961pt) -- +(45.000000:9.899495pt and 6.363961pt) -- +(135.000000:9.899495pt and 6.363961pt) -- +(225.000000:9.899495pt and 6.363961pt) -- cycle;
\clip (50.000000, -0.000000) +(-45.000000:9.899495pt and 6.363961pt) -- +(45.000000:9.899495pt and 6.363961pt) -- +(135.000000:9.899495pt and 6.363961pt) -- +(225.000000:9.899495pt and 6.363961pt) -- cycle;
\draw (50.000000, -0.000000) node {$C_8$};
\end{scope}
\begin{scope}
\draw[fill=white] (50.000000, 85.000000) +(-45.000000:9.899495pt and 6.363961pt) -- +(45.000000:9.899495pt and 6.363961pt) -- +(135.000000:9.899495pt and 6.363961pt) -- +(225.000000:9.899495pt and 6.363961pt) -- cycle;
\clip (50.000000, 85.000000) +(-45.000000:9.899495pt and 6.363961pt) -- +(45.000000:9.899495pt and 6.363961pt) -- +(135.000000:9.899495pt and 6.363961pt) -- +(225.000000:9.899495pt and 6.363961pt) -- cycle;
\draw (50.000000, 85.000000) node {$C_8$};
\end{scope}
\begin{scope}[color=white]
\begin{scope}[color=white]
\begin{scope}
\draw[fill=white] (61.000000, 54.500000) +(-45.000000:0.000000pt and 6.363961pt) -- +(45.000000:0.000000pt and 6.363961pt) -- +(135.000000:0.000000pt and 6.363961pt) -- +(225.000000:0.000000pt and 6.363961pt) -- cycle;
\clip (61.000000, 54.500000) +(-45.000000:0.000000pt and 6.363961pt) -- +(45.000000:0.000000pt and 6.363961pt) -- +(135.000000:0.000000pt and 6.363961pt) -- +(225.000000:0.000000pt and 6.363961pt) -- cycle;
\draw (61.000000, 54.500000) node {{}};
\end{scope}
\end{scope}
\end{scope}
\begin{scope}[color=white]
\begin{scope}[color=white]
\begin{scope}
\draw[fill=white] (65.000000, 54.500000) +(-45.000000:0.000000pt and 6.363961pt) -- +(45.000000:0.000000pt and 6.363961pt) -- +(135.000000:0.000000pt and 6.363961pt) -- +(225.000000:0.000000pt and 6.363961pt) -- cycle;
\clip (65.000000, 54.500000) +(-45.000000:0.000000pt and 6.363961pt) -- +(45.000000:0.000000pt and 6.363961pt) -- +(135.000000:0.000000pt and 6.363961pt) -- +(225.000000:0.000000pt and 6.363961pt) -- cycle;
\draw (65.000000, 54.500000) node {{}};
\end{scope}
\end{scope}
\end{scope}
\begin{scope}[color=white]
\begin{scope}[color=white]
\begin{scope}
\draw[fill=white] (69.000000, 54.500000) +(-45.000000:0.000000pt and 6.363961pt) -- +(45.000000:0.000000pt and 6.363961pt) -- +(135.000000:0.000000pt and 6.363961pt) -- +(225.000000:0.000000pt and 6.363961pt) -- cycle;
\clip (69.000000, 54.500000) +(-45.000000:0.000000pt and 6.363961pt) -- +(45.000000:0.000000pt and 6.363961pt) -- +(135.000000:0.000000pt and 6.363961pt) -- +(225.000000:0.000000pt and 6.363961pt) -- cycle;
\draw (69.000000, 54.500000) node {{}};
\end{scope}
\end{scope}
\end{scope}
\draw (80.000000,85.000000) -- (80.000000,36.000000);
\begin{scope}
\draw[fill=white] (80.000000, 36.000000) +(-45.000000:9.899495pt and 6.363961pt) -- +(45.000000:9.899495pt and 6.363961pt) -- +(135.000000:9.899495pt and 6.363961pt) -- +(225.000000:9.899495pt and 6.363961pt) -- cycle;
\clip (80.000000, 36.000000) +(-45.000000:9.899495pt and 6.363961pt) -- +(45.000000:9.899495pt and 6.363961pt) -- +(135.000000:9.899495pt and 6.363961pt) -- +(225.000000:9.899495pt and 6.363961pt) -- cycle;
\draw (80.000000, 36.000000) node {$P_1$};
\end{scope}
\begin{scope}
\draw[fill=white] (80.000000, 85.000000) +(-45.000000:9.899495pt and 6.363961pt) -- +(45.000000:9.899495pt and 6.363961pt) -- +(135.000000:9.899495pt and 6.363961pt) -- +(225.000000:9.899495pt and 6.363961pt) -- cycle;
\clip (80.000000, 85.000000) +(-45.000000:9.899495pt and 6.363961pt) -- +(45.000000:9.899495pt and 6.363961pt) -- +(135.000000:9.899495pt and 6.363961pt) -- +(225.000000:9.899495pt and 6.363961pt) -- cycle;
\draw (80.000000, 85.000000) node {$P_1$};
\end{scope}
\draw (94.000000,73.000000) -- (94.000000,24.000000);
\begin{scope}
\draw[fill=white] (94.000000, 24.000000) +(-45.000000:9.899495pt and 6.363961pt) -- +(45.000000:9.899495pt and 6.363961pt) -- +(135.000000:9.899495pt and 6.363961pt) -- +(225.000000:9.899495pt and 6.363961pt) -- cycle;
\clip (94.000000, 24.000000) +(-45.000000:9.899495pt and 6.363961pt) -- +(45.000000:9.899495pt and 6.363961pt) -- +(135.000000:9.899495pt and 6.363961pt) -- +(225.000000:9.899495pt and 6.363961pt) -- cycle;
\draw (94.000000, 24.000000) node {$P_2$};
\end{scope}
\begin{scope}
\draw[fill=white] (94.000000, 73.000000) +(-45.000000:9.899495pt and 6.363961pt) -- +(45.000000:9.899495pt and 6.363961pt) -- +(135.000000:9.899495pt and 6.363961pt) -- +(225.000000:9.899495pt and 6.363961pt) -- cycle;
\clip (94.000000, 73.000000) +(-45.000000:9.899495pt and 6.363961pt) -- +(45.000000:9.899495pt and 6.363961pt) -- +(135.000000:9.899495pt and 6.363961pt) -- +(225.000000:9.899495pt and 6.363961pt) -- cycle;
\draw (94.000000, 73.000000) node {$P_2$};
\end{scope}
\draw (108.000000,61.000000) -- (108.000000,12.000000);
\begin{scope}
\draw[fill=white] (108.000000, 12.000000) +(-45.000000:9.899495pt and 6.363961pt) -- +(45.000000:9.899495pt and 6.363961pt) -- +(135.000000:9.899495pt and 6.363961pt) -- +(225.000000:9.899495pt and 6.363961pt) -- cycle;
\clip (108.000000, 12.000000) +(-45.000000:9.899495pt and 6.363961pt) -- +(45.000000:9.899495pt and 6.363961pt) -- +(135.000000:9.899495pt and 6.363961pt) -- +(225.000000:9.899495pt and 6.363961pt) -- cycle;
\draw (108.000000, 12.000000) node {$P_3$};
\end{scope}
\begin{scope}
\draw[fill=white] (108.000000, 61.000000) +(-45.000000:9.899495pt and 6.363961pt) -- +(45.000000:9.899495pt and 6.363961pt) -- +(135.000000:9.899495pt and 6.363961pt) -- +(225.000000:9.899495pt and 6.363961pt) -- cycle;
\clip (108.000000, 61.000000) +(-45.000000:9.899495pt and 6.363961pt) -- +(45.000000:9.899495pt and 6.363961pt) -- +(135.000000:9.899495pt and 6.363961pt) -- +(225.000000:9.899495pt and 6.363961pt) -- cycle;
\draw (108.000000, 61.000000) node {$P_3$};
\end{scope}
\draw (122.000000,48.000000) -- (122.000000,0.000000);
\begin{scope}
\draw[fill=white] (122.000000, -0.000000) +(-45.000000:9.899495pt and 6.363961pt) -- +(45.000000:9.899495pt and 6.363961pt) -- +(135.000000:9.899495pt and 6.363961pt) -- +(225.000000:9.899495pt and 6.363961pt) -- cycle;
\clip (122.000000, -0.000000) +(-45.000000:9.899495pt and 6.363961pt) -- +(45.000000:9.899495pt and 6.363961pt) -- +(135.000000:9.899495pt and 6.363961pt) -- +(225.000000:9.899495pt and 6.363961pt) -- cycle;
\draw (122.000000, -0.000000) node {$P_4$};
\end{scope}
\begin{scope}
\draw[fill=white] (122.000000, 48.000000) +(-45.000000:9.899495pt and 6.363961pt) -- +(45.000000:9.899495pt and 6.363961pt) -- +(135.000000:9.899495pt and 6.363961pt) -- +(225.000000:9.899495pt and 6.363961pt) -- cycle;
\clip (122.000000, 48.000000) +(-45.000000:9.899495pt and 6.363961pt) -- +(45.000000:9.899495pt and 6.363961pt) -- +(135.000000:9.899495pt and 6.363961pt) -- +(225.000000:9.899495pt and 6.363961pt) -- cycle;
\draw (122.000000, 48.000000) node {$P_4$};
\end{scope}
\begin{scope}[color=white]
\begin{scope}[color=white]
\begin{scope}
\draw[fill=white] (133.000000, 54.500000) +(-45.000000:0.000000pt and 6.363961pt) -- +(45.000000:0.000000pt and 6.363961pt) -- +(135.000000:0.000000pt and 6.363961pt) -- +(225.000000:0.000000pt and 6.363961pt) -- cycle;
\clip (133.000000, 54.500000) +(-45.000000:0.000000pt and 6.363961pt) -- +(45.000000:0.000000pt and 6.363961pt) -- +(135.000000:0.000000pt and 6.363961pt) -- +(225.000000:0.000000pt and 6.363961pt) -- cycle;
\draw (133.000000, 54.500000) node {{}};
\end{scope}
\end{scope}
\end{scope}
\begin{scope}[color=white]
\begin{scope}[color=white]
\begin{scope}
\draw[fill=white] (137.000000, 54.500000) +(-45.000000:0.000000pt and 6.363961pt) -- +(45.000000:0.000000pt and 6.363961pt) -- +(135.000000:0.000000pt and 6.363961pt) -- +(225.000000:0.000000pt and 6.363961pt) -- cycle;
\clip (137.000000, 54.500000) +(-45.000000:0.000000pt and 6.363961pt) -- +(45.000000:0.000000pt and 6.363961pt) -- +(135.000000:0.000000pt and 6.363961pt) -- +(225.000000:0.000000pt and 6.363961pt) -- cycle;
\draw (137.000000, 54.500000) node {{}};
\end{scope}
\end{scope}
\end{scope}
\begin{scope}[color=white]
\begin{scope}[color=white]
\begin{scope}
\draw[fill=white] (141.000000, 54.500000) +(-45.000000:0.000000pt and 6.363961pt) -- +(45.000000:0.000000pt and 6.363961pt) -- +(135.000000:0.000000pt and 6.363961pt) -- +(225.000000:0.000000pt and 6.363961pt) -- cycle;
\clip (141.000000, 54.500000) +(-45.000000:0.000000pt and 6.363961pt) -- +(45.000000:0.000000pt and 6.363961pt) -- +(135.000000:0.000000pt and 6.363961pt) -- +(225.000000:0.000000pt and 6.363961pt) -- cycle;
\draw (141.000000, 54.500000) node {{}};
\end{scope}
\end{scope}
\end{scope}
\draw (152.000000,36.000000) -- (152.000000,24.000000);
\begin{scope}
\draw[fill=white] (152.000000, 30.000000) +(-45.000000:9.899495pt and 14.849242pt) -- +(45.000000:9.899495pt and 14.849242pt) -- +(135.000000:9.899495pt and 14.849242pt) -- +(225.000000:9.899495pt and 14.849242pt) -- cycle;
\clip (152.000000, 30.000000) +(-45.000000:9.899495pt and 14.849242pt) -- +(45.000000:9.899495pt and 14.849242pt) -- +(135.000000:9.899495pt and 14.849242pt) -- +(225.000000:9.899495pt and 14.849242pt) -- cycle;
\draw (152.000000, 30.000000) node {$C_9$};
\end{scope}
\draw (152.000000,12.000000) -- (152.000000,0.000000);
\begin{scope}
\draw[fill=white] (152.000000, 6.000000) +(-45.000000:9.899495pt and 14.849242pt) -- +(45.000000:9.899495pt and 14.849242pt) -- +(135.000000:9.899495pt and 14.849242pt) -- +(225.000000:9.899495pt and 14.849242pt) -- cycle;
\clip (152.000000, 6.000000) +(-45.000000:9.899495pt and 14.849242pt) -- +(45.000000:9.899495pt and 14.849242pt) -- +(135.000000:9.899495pt and 14.849242pt) -- +(225.000000:9.899495pt and 14.849242pt) -- cycle;
\draw (152.000000, 6.000000) node {$C_{10}$};
\end{scope}
\draw (170.000000,24.000000) -- (170.000000,12.000000);
\begin{scope}
\draw[fill=white] (170.000000, 18.000000) +(-45.000000:9.899495pt and 14.849242pt) -- +(45.000000:9.899495pt and 14.849242pt) -- +(135.000000:9.899495pt and 14.849242pt) -- +(225.000000:9.899495pt and 14.849242pt) -- cycle;
\clip (170.000000, 18.000000) +(-45.000000:9.899495pt and 14.849242pt) -- +(45.000000:9.899495pt and 14.849242pt) -- +(135.000000:9.899495pt and 14.849242pt) -- +(225.000000:9.899495pt and 14.849242pt) -- cycle;
\draw (170.000000, 18.000000) node {$C_{11}$};
\end{scope}
\draw (184.000000,36.000000) -- (184.000000,0.000000);
\begin{scope}
\draw[fill=white] (184.000000, -0.000000) +(-45.000000:9.899495pt and 6.363961pt) -- +(45.000000:9.899495pt and 6.363961pt) -- +(135.000000:9.899495pt and 6.363961pt) -- +(225.000000:9.899495pt and 6.363961pt) -- cycle;
\clip (184.000000, -0.000000) +(-45.000000:9.899495pt and 6.363961pt) -- +(45.000000:9.899495pt and 6.363961pt) -- +(135.000000:9.899495pt and 6.363961pt) -- +(225.000000:9.899495pt and 6.363961pt) -- cycle;
\draw (184.000000, -0.000000) node {$C_{12}$};
\end{scope}
\begin{scope}
\draw[fill=white] (184.000000, 36.000000) +(-45.000000:9.899495pt and 6.363961pt) -- +(45.000000:9.899495pt and 6.363961pt) -- +(135.000000:9.899495pt and 6.363961pt) -- +(225.000000:9.899495pt and 6.363961pt) -- cycle;
\clip (184.000000, 36.000000) +(-45.000000:9.899495pt and 6.363961pt) -- +(45.000000:9.899495pt and 6.363961pt) -- +(135.000000:9.899495pt and 6.363961pt) -- +(225.000000:9.899495pt and 6.363961pt) -- cycle;
\draw (184.000000, 36.000000) node {$C_{12}$};
\end{scope}
\begin{scope}[color=white]
\begin{scope}[color=white]
\begin{scope}
\draw[fill=white] (195.000000, 54.500000) +(-45.000000:0.000000pt and 6.363961pt) -- +(45.000000:0.000000pt and 6.363961pt) -- +(135.000000:0.000000pt and 6.363961pt) -- +(225.000000:0.000000pt and 6.363961pt) -- cycle;
\clip (195.000000, 54.500000) +(-45.000000:0.000000pt and 6.363961pt) -- +(45.000000:0.000000pt and 6.363961pt) -- +(135.000000:0.000000pt and 6.363961pt) -- +(225.000000:0.000000pt and 6.363961pt) -- cycle;
\draw (195.000000, 54.500000) node {{}};
\end{scope}
\end{scope}
\end{scope}
\begin{scope}[color=white]
\begin{scope}[color=white]
\begin{scope}
\draw[fill=white] (199.000000, 54.500000) +(-45.000000:0.000000pt and 6.363961pt) -- +(45.000000:0.000000pt and 6.363961pt) -- +(135.000000:0.000000pt and 6.363961pt) -- +(225.000000:0.000000pt and 6.363961pt) -- cycle;
\clip (199.000000, 54.500000) +(-45.000000:0.000000pt and 6.363961pt) -- +(45.000000:0.000000pt and 6.363961pt) -- +(135.000000:0.000000pt and 6.363961pt) -- +(225.000000:0.000000pt and 6.363961pt) -- cycle;
\draw (199.000000, 54.500000) node {{}};
\end{scope}
\end{scope}
\end{scope}
\begin{scope}[color=white]
\begin{scope}[color=white]
\begin{scope}
\draw[fill=white] (203.000000, 54.500000) +(-45.000000:0.000000pt and 6.363961pt) -- +(45.000000:0.000000pt and 6.363961pt) -- +(135.000000:0.000000pt and 6.363961pt) -- +(225.000000:0.000000pt and 6.363961pt) -- cycle;
\clip (203.000000, 54.500000) +(-45.000000:0.000000pt and 6.363961pt) -- +(45.000000:0.000000pt and 6.363961pt) -- +(135.000000:0.000000pt and 6.363961pt) -- +(225.000000:0.000000pt and 6.363961pt) -- cycle;
\draw (203.000000, 54.500000) node {{}};
\end{scope}
\end{scope}
\end{scope}
\draw (214.000000,36.000000) -- (214.000000,12.000000);
\begin{scope}
\draw[fill=white] (214.000000, 12.000000) +(-45.000000:9.899495pt and 6.363961pt) -- +(45.000000:9.899495pt and 6.363961pt) -- +(135.000000:9.899495pt and 6.363961pt) -- +(225.000000:9.899495pt and 6.363961pt) -- cycle;
\clip (214.000000, 12.000000) +(-45.000000:9.899495pt and 6.363961pt) -- +(45.000000:9.899495pt and 6.363961pt) -- +(135.000000:9.899495pt and 6.363961pt) -- +(225.000000:9.899495pt and 6.363961pt) -- cycle;
\draw (214.000000, 12.000000) node {$P_5$};
\end{scope}
\begin{scope}
\draw[fill=white] (214.000000, 36.000000) +(-45.000000:9.899495pt and 6.363961pt) -- +(45.000000:9.899495pt and 6.363961pt) -- +(135.000000:9.899495pt and 6.363961pt) -- +(225.000000:9.899495pt and 6.363961pt) -- cycle;
\clip (214.000000, 36.000000) +(-45.000000:9.899495pt and 6.363961pt) -- +(45.000000:9.899495pt and 6.363961pt) -- +(135.000000:9.899495pt and 6.363961pt) -- +(225.000000:9.899495pt and 6.363961pt) -- cycle;
\draw (214.000000, 36.000000) node {$P_5$};
\end{scope}
\draw (228.000000,24.000000) -- (228.000000,0.000000);
\begin{scope}
\draw[fill=white] (228.000000, 24.000000) +(-45.000000:9.899495pt and 6.363961pt) -- +(45.000000:9.899495pt and 6.363961pt) -- +(135.000000:9.899495pt and 6.363961pt) -- +(225.000000:9.899495pt and 6.363961pt) -- cycle;
\clip (228.000000, 24.000000) +(-45.000000:9.899495pt and 6.363961pt) -- +(45.000000:9.899495pt and 6.363961pt) -- +(135.000000:9.899495pt and 6.363961pt) -- +(225.000000:9.899495pt and 6.363961pt) -- cycle;
\draw (228.000000, 24.000000) node {$P_6$};
\end{scope}
\begin{scope}
\draw[fill=white] (228.000000, -0.000000) +(-45.000000:9.899495pt and 6.363961pt) -- +(45.000000:9.899495pt and 6.363961pt) -- +(135.000000:9.899495pt and 6.363961pt) -- +(225.000000:9.899495pt and 6.363961pt) -- cycle;
\clip (228.000000, -0.000000) +(-45.000000:9.899495pt and 6.363961pt) -- +(45.000000:9.899495pt and 6.363961pt) -- +(135.000000:9.899495pt and 6.363961pt) -- +(225.000000:9.899495pt and 6.363961pt) -- cycle;
\draw (228.000000, -0.000000) node {$P_6$};
\end{scope}
\begin{scope}[color=white]
\begin{scope}[color=white]
\begin{scope}
\draw[fill=white] (239.000000, 54.500000) +(-45.000000:0.000000pt and 6.363961pt) -- +(45.000000:0.000000pt and 6.363961pt) -- +(135.000000:0.000000pt and 6.363961pt) -- +(225.000000:0.000000pt and 6.363961pt) -- cycle;
\clip (239.000000, 54.500000) +(-45.000000:0.000000pt and 6.363961pt) -- +(45.000000:0.000000pt and 6.363961pt) -- +(135.000000:0.000000pt and 6.363961pt) -- +(225.000000:0.000000pt and 6.363961pt) -- cycle;
\draw (239.000000, 54.500000) node {{}};
\end{scope}
\end{scope}
\end{scope}
\begin{scope}[color=white]
\begin{scope}[color=white]
\begin{scope}
\draw[fill=white] (243.000000, 54.500000) +(-45.000000:0.000000pt and 6.363961pt) -- +(45.000000:0.000000pt and 6.363961pt) -- +(135.000000:0.000000pt and 6.363961pt) -- +(225.000000:0.000000pt and 6.363961pt) -- cycle;
\clip (243.000000, 54.500000) +(-45.000000:0.000000pt and 6.363961pt) -- +(45.000000:0.000000pt and 6.363961pt) -- +(135.000000:0.000000pt and 6.363961pt) -- +(225.000000:0.000000pt and 6.363961pt) -- cycle;
\draw (243.000000, 54.500000) node {{}};
\end{scope}
\end{scope}
\end{scope}
\begin{scope}[color=white]
\begin{scope}[color=white]
\begin{scope}
\draw[fill=white] (247.000000, 54.500000) +(-45.000000:0.000000pt and 6.363961pt) -- +(45.000000:0.000000pt and 6.363961pt) -- +(135.000000:0.000000pt and 6.363961pt) -- +(225.000000:0.000000pt and 6.363961pt) -- cycle;
\clip (247.000000, 54.500000) +(-45.000000:0.000000pt and 6.363961pt) -- +(45.000000:0.000000pt and 6.363961pt) -- +(135.000000:0.000000pt and 6.363961pt) -- +(225.000000:0.000000pt and 6.363961pt) -- cycle;
\draw (247.000000, 54.500000) node {{}};
\end{scope}
\end{scope}
\end{scope}
\draw (258.000000,12.000000) -- (258.000000,0.000000);
\begin{scope}
\draw[fill=white] (258.000000, 6.000000) +(-45.000000:9.899495pt and 14.849242pt) -- +(45.000000:9.899495pt and 14.849242pt) -- +(135.000000:9.899495pt and 14.849242pt) -- +(225.000000:9.899495pt and 14.849242pt) -- cycle;
\clip (258.000000, 6.000000) +(-45.000000:9.899495pt and 14.849242pt) -- +(45.000000:9.899495pt and 14.849242pt) -- +(135.000000:9.899495pt and 14.849242pt) -- +(225.000000:9.899495pt and 14.849242pt) -- cycle;
\draw (258.000000, 6.000000) node {$U$};
\end{scope}
\begin{scope}[color=white]
\begin{scope}[color=white]
\begin{scope}
\draw[fill=white] (269.000000, 54.500000) +(-45.000000:0.000000pt and 6.363961pt) -- +(45.000000:0.000000pt and 6.363961pt) -- +(135.000000:0.000000pt and 6.363961pt) -- +(225.000000:0.000000pt and 6.363961pt) -- cycle;
\clip (269.000000, 54.500000) +(-45.000000:0.000000pt and 6.363961pt) -- +(45.000000:0.000000pt and 6.363961pt) -- +(135.000000:0.000000pt and 6.363961pt) -- +(225.000000:0.000000pt and 6.363961pt) -- cycle;
\draw (269.000000, 54.500000) node {{}};
\end{scope}
\end{scope}
\end{scope}
\begin{scope}[color=white]
\begin{scope}[color=white]
\begin{scope}
\draw[fill=white] (275.000000, 54.500000) +(-45.000000:2.828427pt and 6.363961pt) -- +(45.000000:2.828427pt and 6.363961pt) -- +(135.000000:2.828427pt and 6.363961pt) -- +(225.000000:2.828427pt and 6.363961pt) -- cycle;
\clip (275.000000, 54.500000) +(-45.000000:2.828427pt and 6.363961pt) -- +(45.000000:2.828427pt and 6.363961pt) -- +(135.000000:2.828427pt and 6.363961pt) -- +(225.000000:2.828427pt and 6.363961pt) -- cycle;
\draw (275.000000, 54.500000) node {{}};
\end{scope}
\end{scope}
\end{scope}
\draw (285.500000, -10.000000) node[below] {$M$};
\draw[fill=white] (281.000000, -4.500000) rectangle (290.000000, 4.500000);
\draw[very thin] (285.500000, 0.450000) arc (90:150:4.500000pt);
\draw[very thin] (285.500000, 0.450000) arc (90:30:4.500000pt);
\draw[->,>=stealth] (285.500000, -4.050000) -- +(80:7.794229pt);
\draw[draw opacity=1.000000,fill opacity=0.200000,color=black,dotted,thick,rounded corners=6pt] (6.000000,95.000000) rectangle (62.000000,-10.000000);
\draw (34.000000, -10.000000) node[below,color=black] {Convolution};
\draw[draw opacity=1.000000,fill opacity=0.200000,color=black,dotted,thick] (68.000000,95.000000) rectangle (134.000000,-10.000000);
\draw (101.000000, -10.000000) node[below,color=black] {Pooling};
\draw[draw opacity=1.000000,fill opacity=0.200000,color=black,dotted,thick,rounded corners=6pt] (140.000000,46.000000) rectangle (196.000000,-10.000000);
\draw[draw opacity=1.000000,fill opacity=0.200000,color=black,dotted,thick] (202.000000,46.000000) rectangle (240.000000,-10.000000);
\draw[draw opacity=1.000000,fill opacity=0.200000,color=black,dashed,thick] (246.000000,22.000000) rectangle (270.000000,-10.000000);
\draw (258.000000, -10.000000) node[text width=40pt,below,text centered,color=black] {Fully \\ Connected};
\draw (199.000000, -10) node[below] {$\cdots$};
\end{tikzpicture}